\documentclass[lettersize,journal]{IEEEtran}
\usepackage{cite}
\usepackage{graphicx}
\usepackage[cmex10]{amsmath}
\usepackage{amssymb}
\usepackage{amsthm}
\usepackage{array}
\usepackage{algorithm,algorithmic}
\usepackage{chemarrow}
\usepackage{multirow}
\usepackage{extarrows}
\usepackage{setspace}
\usepackage{booktabs}
\usepackage{tabulary}
\usepackage[shortlabels]{enumitem}
\usepackage{balance}
\usepackage{subfigure}
\usepackage[cmex10]{amsmath}
\usepackage[shortlabels]{enumitem}
\usepackage{amssymb}
\usepackage{amsthm}
\usepackage{multirow}
\usepackage{cite}
\usepackage{color}
\usepackage{setspace}
\usepackage{ragged2e}
\usepackage{stmaryrd} 
\usepackage{tabularx}
\usepackage{makecell}
\usepackage{footnote}
\usepackage{threeparttable}
\usepackage{eqparbox}
\usepackage{bm}
\usepackage[hyphens]{url}
\makeatletter
\newcommand{\ALOOP}[1]{\ALC@it\algorithmicloop\ #1%
	\begin{ALC@loop}}
	\newcommand{\ENDALOOP}{\end{ALC@loop}\ALC@it\algorithmicendloop}

\newtheorem{theorem}{\textbf{\emph{Theorem}}}
\newtheorem{definition}{\textbf{\emph{Definition}}}

\newcommand{\main}{OblivGM}
\newcommand{\cs}{$\mathcal{CS}_{1}$}
\newcommand{\css}{$\mathcal{CS}_{2}$}
\newcommand{\csss}{$\mathcal{CS}_{3}$}
\newcommand{\csa}{$\mathcal{CS}_{\{1,2,3\}}$}

\newcommand{\cli}{$\mathcal{SF}$}

\newcommand{\yifeng}[1]{\textsf{\color{red}{[{Yifeng: #1}]}}}

\newcommand{\revise}{\textcolor{black}}

\begin{document}
	
	%\title{{\main}: A System for Privacy-Preserving Attributed Subgraph Matching Powered by Function Secret Sharing}
	% \title{{\main}: Oblivious and Encrypted Attributed Subgraph Matching as a Cloud Service}
	
	\title{{\main}: Oblivious Attributed Subgraph Matching as a Cloud Service}
	
	\author{Songlei Wang, Yifeng Zheng, Xiaohua Jia, \IEEEmembership{Fellow, IEEE}, Hejiao Huang, and Cong Wang,  \IEEEmembership{Fellow, IEEE}

	\thanks{Songlei Wang and Yifeng Zheng are with the School of Computer Science and Technology, Harbin Institute of Technology, Shenzhen, Guangdong 518055, China (e-mail: songlei.wang@outlook.com; yifeng.zheng@hit.edu.cn).}

	\thanks{Xiaohua Jia is with the School of Computer Science and Technology, Harbin Institute of Technology, Shenzhen, Guangdong 518055, China, and also with the Department of Computer Science, City University of Hong Kong, Hong Kong, China (e-mail: csjia@cityu.edu.hk).}

	\thanks{Hejiao Huang is with the School of Computer Science and Technology, Harbin Institute of Technology, Shenzhen, Guangdong 518055, China, and also with the Guangdong Provincial Key Laboratory of Novel Security Intelligence Technologies (e-mail: huanghejiao@hit.edu.cn).}

	\thanks{Cong Wang is with the Department of Computer Science, City University of Hong Kong, Hong Kong, China (e-mail: congwang@cityu.edu.hk).}

	\thanks{Corresponding author: Yifeng Zheng.}

	}

	\IEEEtitleabstractindextext{
		\begin{abstract}
			In recent years there has been growing popularity of leveraging cloud computing for storing and querying attributed graphs, which have been widely used to model complex structured data in various applications. Such trend of outsourced graph analytics, however, is accompanied with  critical privacy concerns regarding the information-rich and proprietary attributed graph data. In light of this, we design, implement, and evaluate OblivGM, a new system aimed at oblivious graph analytics services outsourced to the cloud. OblivGM focuses on the support for attributed subgraph matching, one popular and fundamental graph query functionality aiming to retrieve from a large attributed graph subgraphs isomorphic to a small query graph. Built from a delicate synergy of insights from attributed graph modelling and advanced lightweight cryptography, OblivGM protects the confidentiality of data content associated with attributed graphs and queries, conceals the connections among vertices in attributed graphs, and hides search access patterns. Meanwhile, OblivGM flexibly supports oblivious evaluation of varying subgraph queries, which may contain equality and/or range predicates. Extensive experiments over a real-world attributed graph dataset demonstrate that while providing strong security guarantees,  OblivGM achieves practically affordable performance (with query latency on the order of a few seconds).
		\end{abstract}
		
		
	}
	
	\maketitle

	\IEEEdisplaynontitleabstractindextext
	
	\IEEEpeerreviewmaketitle

	\section{Introduction}
	\label{sec:intro}

	Attributed graphs, as one kind of the most popular graph data models \cite{bi2016efficient}, have been widely used to capture the interactions between entities in various applications, such as social networks, financial services, and manufacturing industries \cite{graphApp}. 
	With the widespread adoption of cloud computing \cite{liu2020enabling,zheng2022efficient}, there has been growing popularity of enterprises resorting to commercial clouds as the back-end to store and query their attributed graphs (e.g., \cite{airbnb,PIXNET}, \revise{to list a few}). 
	While the benefits are well-understood, deploying such graph analytics services in the public cloud also poses threats \cite{RenWW12} to the privacy of information-rich attributed graph data and may not be good for the business interests of these enterprises as the graph data is proprietary. 
	Hence, there is an urgent demand that security must be embedded in such cloud-backed graph analytics services from the very beginning, providing protection for the outsourced attributed graphs and queries.

	% so that the cloud-hosted attributed graph data, while being utilized, gains strong protection on its confidentiality. 

	As one of the most fundamental functionalities in querying attributed graphs, attributed subgraph matching, which is the focus in this paper, aims to retrieve from a large attributed graph subgraphs isomorphic to a given small query graph \cite{bi2016efficient}. 
	Attributed subgraph matching is a powerful tool in various applications, such as anti-money laundering \cite{tong2007fast}, chemical compound search \cite{yan2004graph}, and social network analysis \cite{bi2016efficient}.
	A concrete example is that retrieving all users whose ego-networks isomorphic to a given ego-network from a social network \cite{PGQL}. 
	Different from regular subgraph matching which only considers the structure matching \cite{qiao2017subgraph}, attributed subgraph matching is more sophisticated as it additionally considers matching against vertices' attributes and types \cite{bi2016efficient}.

	In the literature, privacy-aware graph query processing has received wide attention in recent years.
	Most of existing works, however, focus on dealing with graph query functionalities that are different from attributed subgraph matching, like privacy-preserving shortest path queries \cite{meng2015grecs,xie2016practical,wang2017secgdb,ghosh2021efficient} and privacy-preserving breadth-first search \cite{blanton2013data,asharov2017privacy,araki2021secure}. 
	Little work has been done for privacy-preserving attributed subgraph matching, where the state-of-the-art protocol under a similar outsourcing scenario is PGP proposed by Huang \textit{et al.} \cite{huang2021privacy}.
	The PGP protocol relies on perturbation techniques---$t$-closeness and $k$-automorphism---to protect the attributed graph and subgraph queries.
	%
	% In a very recent work \cite{huang2021privacy}, Huang \textit{et al.} present PGP, focusing on the support for privacy-preserving attributed subgraph matching in an outsourcing setting, where the attributed graph and subgraph queries are protected by $t$-closeness and $k$-automorphism. 
	%
	Despite being a valuable design point, PGP is not satisfactory for practical use due to the following downsides.
	
	Firstly, the construction of PGP is tailored for limited subgraph queries with equality predicates, which is far from sufficient for practical use.
	Indeed practical attributed subgraph matching systems (such as ORACLE's PGQL \cite{PGQL} and Amazon's Neptune \cite{Amazon}) should flexibly support subgraph queries containing equality predicates as well as range predicates.
	Secondly, PGP relies on the notion of $t$-closeness for protecting attribute values (via generalization), which is not strong in protecting the confidentiality of data content from a cryptographic perspective, as well degrades the quality of matching results.
	Thirdly, PGP does not consider hiding search access patterns \cite{curtmola2006searchable}, which have been shown to be exploitable for various attacks \cite{kornaropoulos2020state,oya2021hiding,damie2021highly} to learn information about queries and database contents.
	Therefore, how to enable privacy-preserving attributed subgraph matching is still challenging and remains to be fully explored.

	\iffalse
	\yifeng{to do}Firstly,  PGP does not provide strong security guarantees. 
	%
	In particular, it only utilizes $t$-closeness and $k$-automorphism instead of strong cryptographic techniques to protect attributed graphs and subgraph queries.
	%
	The use of $t$-closeness and $k$-automorphism results that PGP must tune between matching accuracy and privacy parameters $t, k$.
	%
	Secondly, PGP leaks search access patterns \cite{curtmola2006searchable}.
	%
	The leakage can be exploited by a large number of emerging attacks \cite{kornaropoulos2020state,oya2021hiding,damie2021highly} even if attributed graphs and subgraph queries are protected.
	%
	Thirdly, PGP only supports limited exact subgraph queries, which is far from sufficient for practical attributed subgraph matching systems, such as ORACLE's PGQL \cite{PGQL} and Amazon's Neptune \cite{Amazon}, which require versatile functionalities including range matching and the mix of exact and range matching. 
	%
	Therefore, how to enable practical privacy-preserving attributed subgraph matching is still challenging and remains to be fully explored.
	\fi

	% with support for strong security and  practical functionalities is challenging and remains to be fully explored.

	In light of the above, in this paper, we design, implement, and evaluate {\main}, a new system enabling oblivious attributed subgraph matching services outsourced to the cloud.
	%
	% \textit{the first system providing cryptographic guarantees for secure attributed subgraph matching}. 
	%
	%
	{\main} allows the cloud hosting an outsourced encrypted attributed graph to obliviously provide subgraph matching services, providing protection for the attributed graph, subgraph queries, and query results.
	{\main} protects the confidentiality of data content associated with the attributed graphs and queries, conceals the connections among vertices in the attributed graph, and hides the search access patterns during the subgraph matching process.
	Besides, {\main} supports secure and rich matching functionalities, in which a subgraph query can contain equality predicates and/or range predicates, and oblivious predicate evaluation can be effectively performed at the cloud.
	At a high level, {\main} is built from a delicate synergy of insights from attributed graph modelling and advanced lightweight cryptography (such as replicated secret sharing (RSS), function secret sharing (FSS), and secure shuffle).
	We highlight our contributions below:

	\begin{itemize}
		\item We present {\main}, a new system enabling oblivious attributed subgraph matching services outsourced to the cloud, with stronger security and richer functionalities over prior art.
		
		% \item Compared with the existing works, {\main} provides not only stronger security guarantees---providing cryptographic guarantees and hiding the search access patterns, but also support for richer functionalities---exact subgraph matching, range subgraph matching as well as their mix.
		
		\item We show how to adequately model the attributed graph and subgraph queries to facilitate secure attributed subgraph matching, and devise custom constructions for encryption of the attributed graph and (randomized) generation of secure query tokens.
		
		\item We devise a suite of secure components to support oblivious attributed subgraph matching at the cloud, including secure predicate evaluation over candidate vertices, secure matched vertices fetching, and secure neighboring vertices accessing. 
		
		\item We formally analyze the security of {\main}, make a GPU-accelerated full-fledged prototype implementation, and conduct extensive evaluations over a real-world attributed graph dataset (with 107614 vertices and 13673453 edges). The results demonstrate that while providing strong security guarantees, {\main} has practically affordable performance (with query latency on the order of a few seconds).
	\end{itemize}
	The rest of this paper is organized as follows. Section \ref{sec:related_work} discusses the related work. Section \ref{sec:pre} introduces preliminaries. Section \ref{sec:problem_def} presents the problem statement. Section \ref{sec:main} gives the detail design of {\main}. The security analysis is presented in Section \ref{sec:security_analysis}, followed by the performance evaluation in Section \ref{sec:exp}. Finally, we conclude this paper in Section \ref{sec:conclusion}.
	
	\section{Related Work}
	\label{sec:related_work}
	
	\subsection{Graph Search in the Plaintext Domain}
	
	Graphs have been widely used to model structured data in various applications (such as social networks, financial transactions, and more \cite{wu2021learning}), due to their powerful capabilities of characterizing the complex interactions among entities in the real world.
	As one of the most fundamental functionalities in graph data analytics, graph search has gained wide attention and various algorithms have been proposed for handling different queries on graphs, e.g., subgraph matching \cite{bi2016efficient}, graph similarity search \cite{kim2021boosting}, graph keyword search \cite{jiang2019generic}, breadth-first search \cite{hu2019direction}, and shortest path search \cite{sommer2014shortest}.
	However, all of them consider the execution of graph search in the plaintext domain without privacy protection.
	
	\subsection{Privacy-Aware Graph Query Processing}
	\label{sec:related_work_2}
	In recent years, great efforts have been devoted to advancing privacy-preserving graph search.
	Chase \textit{et al.} \cite{chase2010structured} propose structured encryption under the searchable encryption framework to support neighboring vertices queries. Privacy-preserving shortest path search \cite{meng2015grecs,xie2016practical,wang2017secgdb,ghosh2021efficient} and privacy-preserving breadth-first search \cite{blanton2013data,asharov2017privacy,araki2021secure} have also gained wide attention. 
	Another line of work studies privacy-preserving subgraph matching, which is much more challenging because more complex operations are required in the ciphertext domain. 
	Some works \cite{ding2019novel,zuo2020privacy} only consider structure matching and work on unattributed graphs.
	Others \cite{chang2016privacy,huang2021privacy,cao2011privacy,fan2015asymmetric,xu2021privacy} study privacy-aware attributed subgraph matching, which is more sophisticated as it additionally considers the matching against vertices' attributes and types. 
	The works \cite{cao2011privacy,fan2015asymmetric,xu2021privacy} target application scenarios different from ours. 
	Specifically, the works \cite{cao2011privacy,fan2015asymmetric} focus on graph containment query, i.e., given a query graph $q$ and an attributed graph $\mathcal{G}$, they just aim to output \textit{whether} $q$ is a subgraph isomorphic to $\mathcal{G}$ or not, while the work \cite{xu2021privacy} considers publicly known attributed graphs.

	The works that are most related to ours are \cite{chang2016privacy,huang2021privacy}, which aim to securely retrieve from an outsourced attributed graph subgraphs isomorphic to a given query graph, with protection for both the attributed graph and query. 
	The state-of-the-art design is PGP \cite{huang2021privacy}, which makes use of $k$-automorphism for obfuscating graph structure and $t$-closeness for generalizing attribute values.
	As mentioned above, PGP is subject to several crucial downsides in terms of security and functionality, which greatly limit its practical usability. 
	%
	% In light of this gap, we present the first system design {\main}. 
	%
	Compared to PGP, {\main} is much advantageous in that it provides much stronger security, supports richer matching functionalities, and does not rely on parameter tuning for accuracy.

	\section{Preliminaries}
		\label{sec:pre}
	\subsection{Attributed Subgraph Matching}

	%As one kind of the most popular graph data models, attributed graphs have been widely used to model complex heterogeneous data in various applications, such as social networks, financial transactions and bill of materials \cite{graphApp}, 
	
	In attributed graphs, \textit{vertices} represent entities and \textit{edges} represent the connections between entities. 
	Attributed graphs are usually heterogeneous, i.e., vertices and edges are of different types and vertices also have different attributes.
	Attributed graphs can be formally defined as follows \cite{bi2016efficient}.
	
	\begin{definition}
		An attributed graph is defined as $\mathcal{G}=\{\mathcal{V},\mathcal{E},\mathcal{T},  \mathcal{A}\}$, where (1) $\mathcal{V}=\{\mathtt{V}_{1},\cdots,\mathtt{V}_{N}\}$ is a set of $N$ vertices; (2) $\mathcal{E}=\{e_{i,j}=(\mathtt{V}_{i},\mathtt{V}_{j}): 1\leq i,j\leq N,i\neq j\}$ is a set of edges; (3) $\mathcal{T}$ is a set of types and each vertex or edge has and only has one type; (4) $ \mathcal{A}$ is a set of vertex attributes and each vertex has one or more attributes. 
	\end{definition}

	Given an attributed graph $\mathcal{G}$ and a subgraph query $q$, attributed subgraph matching is to retrieve all subgraphs  $\{g_{m}\}$ isomorphic to $q$ from  $\mathcal{G}$. 
	Prior works \cite{chang2016privacy,huang2021privacy} on privacy-preserving attributed subgraph matching give the formal definition of graph isomorphism as that in \cite{bi2016efficient}, but we note that they only focus on exact graph matching. 
	Actually, practical attributed subgraph matching systems (e.g., ORACLE's PGQL \cite{PGQL} and Amazon's Neptune \cite{Amazon}) should support not only exact matching but also \textit{range matching}, i.e., ``$\mathtt{where}$'' in structured query language. 
	Therefore, on the basis of the definitions in \cite{bi2016efficient}, we give the more advanced graph isomorphism definition {\main} focuses on as follows. 
	\begin{definition}
		Given a subgraph $g=\{\mathcal{V}_{g},\mathcal{E}_{g}\}$ in the attributed graph $\mathcal{G}$ and a query graph $q=\{\mathcal{V}_{q},\mathcal{E}_{q}\}$, $g$ is isomorphic to $q$, if and only if there exists a bijective function $f: \mathcal{V}_{g}\to \mathcal{V}_{q}, \mathcal{E}_{g}\to \mathcal{E}_{q}$ such that 1) $\forall \mathtt{V}_{i}\in \mathcal{V}_{g}, f(\mathtt{V}_{i})\in  \mathcal{V}_{q}\Rightarrow \mathtt{T}(\mathtt{V}_{i})=\mathtt{T}(f(\mathtt{V}_{i}))$ and $\mathtt{Att}(\mathtt{V}_{i})=\mathtt{Att}(f(\mathtt{V}_{i}))$ or $\mathtt{Att}(\mathtt{V}_{i})\in\mathtt{Att}(f(\mathtt{V}_{i}))$; 2) $\forall e_{i,j}\in\mathcal{E}_{g}, f(e_{i,j})\in \mathcal{E}_{q} \Rightarrow \mathtt{T}(e_{i,j})=\mathtt{T}(f(e_{i,j}))$, where $\mathtt{T}(\cdot)$ and $\mathtt{Att}(\cdot)$ represent the type and attribute of $\cdot$, respectively.
	\end{definition}
	
	Note that the only difference is that we modify ``$\mathtt{Att}(\mathtt{V}_{i})=\mathtt{Att}(f(\mathtt{V}_{i}))$'' in the definition of graph isomorphism in \cite{bi2016efficient,chang2016privacy,huang2021privacy} into ``$\mathtt{Att}(\mathtt{V}_{i})=\mathtt{Att}(f(\mathtt{V}_{i}))$ or $\mathtt{Att}(\mathtt{V}_{i})\in\mathtt{Att}(f(\mathtt{V}_{i}))$''.
	Specifically, prior works \cite{chang2016privacy,huang2021privacy} only consider exact matching (i.e., equality predicate), where the attribute of each vertex in the query $q$ is associated with an ``exact value'', and the matching is defined as that each vertex in the subgraph $g$ has an attribute value equal to the corresponding value in $q$, i.e., $\mathtt{Att}(\mathtt{V}_{i})=\mathtt{Att}(f(\mathtt{V}_{i}))$. 
	{\main} considers not only exact matching like \cite{chang2016privacy,huang2021privacy} but also range matching (i.e., range predicate), where the attribute of each vertex in $q$ is associated with a ``range'' (single-sided or an interval), and the matching is defined as that each vertex in $g$ has an attribute value \textit{within} the corresponding range in $q$, i.e., $\mathtt{Att}(\mathtt{V}_{i})\in\mathtt{Att}(f(\mathtt{V}_{i}))$. 
	In addition, {\main} also considers and flexibly supports mixed matching, where some attributes are associated with exact values, while others are associated with ranges.
	%
	% We name the value associated with each attribute in queries as \textit{predicate}, including equality predicates or range predicate corresponding to exact matching and range matching, respectively.\yifeng{maybe we do not need to mention this here}

	\begin{figure}
		\centering
		\includegraphics[width=0.9\linewidth]{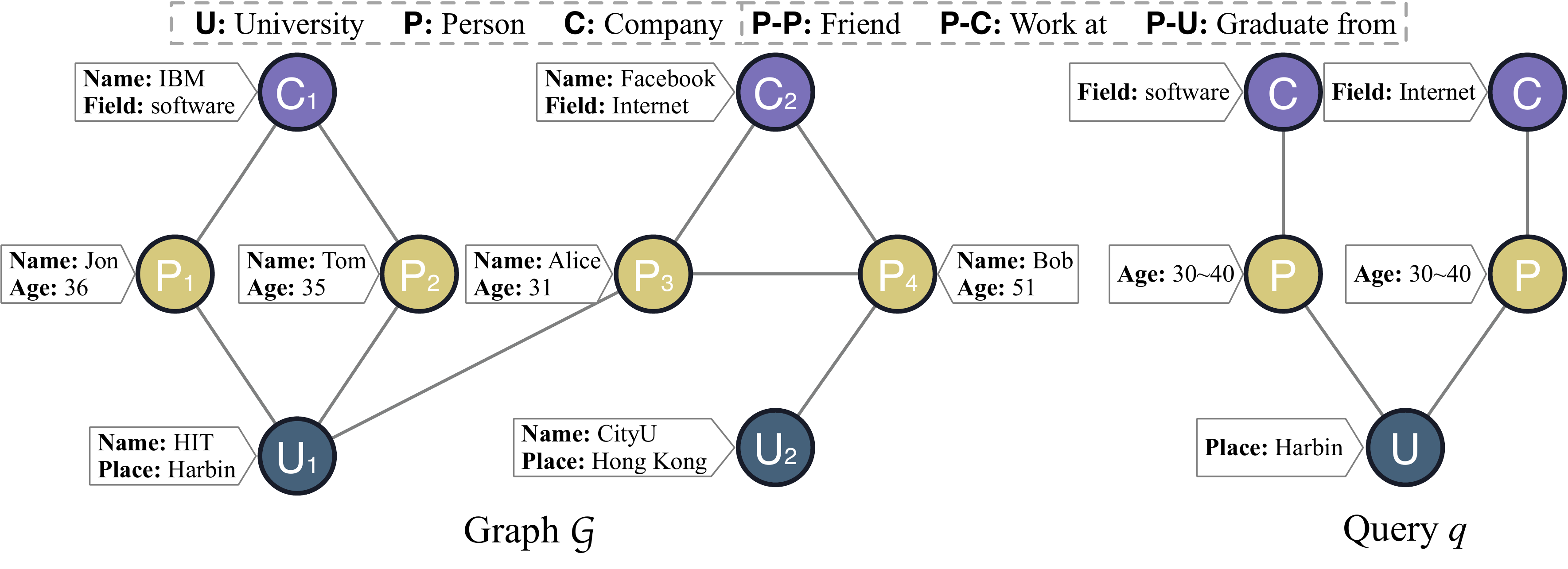}
		\caption{Illustration of attributed subgraph matching.}
		\label{fig:viewgraph}
		\vspace{-15pt}
	\end{figure}

	% To make the problem of attributed subgraph matching we target more concrete,
	
	For clarity, we illustrate an attributed graph $\mathcal{G}$ and a query $q$ in Fig. \ref{fig:viewgraph}. 
	$\mathcal{G}$ has three types of vertices or entities, i.e., ``university ($\mathtt{U}$)'', ``person ($\mathtt{P}$)'', and ``company ($\mathtt{C}$)''. 
	The connection between different vertices implies the edge type, such as ``friend ($\mathtt{P}$-$\mathtt{P}$)'', ``work at ($\mathtt{P}$-$\mathtt{C}$)'' and ``graduate from ($\mathtt{P}$-$\mathtt{U}$)''. 
	%
	%The vertices with the same type have the same attribute types but different attribute values, such as all person vertices have age attribute but their ages are different. 
	%
	The query $q$ represents that a user wants to retrieve two persons satisfying the following conditions: (1) both of them graduated from the same university located in ``Harbin''; (2) their ages are within $[30,40]$; (3) one of them is working at a software company and the other one is working at an Internet company. 
	Then the final matching results consist of two subgraphs:
	\begin{equation}\notag
		g_{1}:=\mathtt{U}_{1}
		\begin{array}{cccc}
			\nearrow\mathtt{P}_{1}\to\mathtt{C}_{1}\\
			\searrow \mathtt{P}_{3}\to\mathtt{C}_{2}
		\end{array},~
		g_{2}:=\mathtt{U}_{1}
		\begin{array}{cccc}
			\nearrow	\mathtt{P}_{2}\to\mathtt{C}_{1}\\
			\searrow \mathtt{P}_{3}\to\mathtt{C}_{2}
		\end{array} .
	\end{equation}
	
	\subsection{Replicated Secret Sharing}
	\label{sec:rss}
	
	Given a secret bit $x\in\mathbb{Z}_{2}$, replicated secret sharing (RSS) \cite{araki2016high} splits it into three shares $\langle x\rangle_{1}$, $ \langle x\rangle_{2}$ and $ \langle x\rangle_{3}\in\mathbb{Z}_{2}$, where $x=\langle x\rangle_{1}\oplus \langle x\rangle_{2}\oplus \langle x\rangle_{3}$.
	Three pairs of shares $(\langle x\rangle_{1}, \langle x\rangle_{2})$, $(\langle x\rangle_{2}, \langle x\rangle_{3})$ and $(\langle x\rangle_{3}, \langle x\rangle_{1})$ are held respectively by three parties $P_{1}$, $P_{2}$ and $P_{3}$, where $P_{i}$ holds the $i$-th pair.
	%
	% are held by three parties, $P_{1}$, $P_{2}$ and $P_{3}$, respectively  \cite{araki2016high}. 
	%
	For the ease of presentation, we write $i\pm 1$ to represent the next (+) party (or secret share) or previous (-) party (or secret share) with wrap around, i.e., $P_{3+1}$ (or $\langle x\rangle_{3+1}$) is $P_{1}$ (or $\langle x\rangle_{1}$) and $P_{1-1}$ (or $\langle x\rangle_{1-1}$) is $P_{3}$ (or $\langle x\rangle_{3}$).
	With this, we use $(\langle x\rangle_{i}, \langle x\rangle_{i+1})$ to represent the shares held by $P_{i}$ ($i\in\{1,2,3\}$) and denote such a sharing of $x$ as $\llbracket x\rrbracket$.
	
	The basic operations in the binary RSS domain are as follows. 
	(1) \textit{XOR $\oplus$}. XOR operations on secret-shared bits only require local computation. 
	To compute $\llbracket u \rrbracket=\llbracket x\oplus y \rrbracket$, each $P_{i}$ locally computes $\langle u\rangle_{i}=\langle x\rangle_{i}\oplus\langle y\rangle_{i}$ and $\langle u\rangle_{i+1}=\langle x\rangle_{i+1}\oplus\langle y\rangle_{i+1}$. 
	(2) \textit{AND $\otimes$}. 
	To compute $\llbracket z \rrbracket=\llbracket x\otimes y \rrbracket$, each $P_{i}$ first locally computes $\langle z\rangle_{i}=\langle x\rangle_{i}\otimes\langle y\rangle_{i}\oplus\langle x\rangle_{i}\otimes\langle y\rangle_{i+1}\oplus\langle x\rangle_{i+1}\otimes\langle y\rangle_{i}$. 
	This will generate a 3-out-of-3 additive secret sharing of $z$ among the three parties, i.e., each $P_i$ only holds $\langle z\rangle_{i}$.
	In order to obtain a RSS of $z$ for subsequent computation, a re-sharing operation can be performed as follows.
	Each $P_{i}$ sends $P_{i+1}$ a blinded share $\langle z\rangle_{i}\oplus \langle \alpha \rangle_{i}$, where $\langle \alpha \rangle_{i}$ is a share from a fresh secret sharing of $0$, i.e., $\langle \alpha \rangle_{1}\oplus \langle \alpha \rangle_{2} \oplus \langle \alpha \rangle_{3}=0$.
	Such fresh secret sharing of $0$ can be efficiently generated based on a pseudorandom function (PRF) $F$ with output domain $\mathbb{Z}_{2}$.
	In particular, in an initialization phase, each $P_{i}$ samples a PRF key $k_{i}$ and sends $k_{i}$ to $P_{i+1}$. 
	To generate the share $\langle \alpha \rangle_{i}$ for the $j$-th fresh secret sharing of $0$, $P_{i}$ computes $\langle \alpha \rangle_{i}=F(k_{i},j)\oplus F(k_{i-1},j)$, which satisfies $\langle \alpha \rangle_{1}\oplus\langle \alpha \rangle_{2}\oplus\langle \alpha \rangle_{3}=0$.

	\iffalse
	Then $P_{i}$ locally draws the share $\langle \alpha \rangle_{i}$ of the $j_{th}$ fresh secret-sharing of 0 by $\langle \alpha \rangle_{i}=F(k_{i},j)\oplus F(k_{i-1},j)$, which satisfies $\langle \alpha \rangle_{1}\oplus\langle \alpha \rangle_{2}\oplus\langle \alpha \rangle_{3}=0$.

	$(\langle \alpha \rangle_{1},\langle \alpha \rangle_{2},\langle \alpha \rangle_{3}$ is a fresh secret-sharing of $0$ and can be derived from a pseudorandom function (PRF). 
	%
	A method to allow $P_{\{1,2,3\}}$ to efficiently generate a fresh secret-sharing of 0 is as follows.
	%
	Let $F$ be a PRF with output domain $\mathbb{Z}_{2}$, in the initialization phase, each $P_{i}$ samples a PRF key $k_{i}$ and sends $k_{i}$ to $P_{i+1}$. 
	%
	Then $P_{i}$ locally draws the share $\langle \alpha \rangle_{i}$ of the $j_{th}$ fresh secret-sharing of 0 by $\langle \alpha \rangle_{i}=F(k_{i},j)\oplus F(k_{i-1},j)$, which satisfies $\langle \alpha \rangle_{1}\oplus\langle \alpha \rangle_{2}\oplus\langle \alpha \rangle_{3}=0$.
	\fi

	%
	
	%The basic operations in arithmetic RSS domain are similar to that in binary RSS. 
	%
	%In particular, the XOR operations are replaced by the addition operations and the AND operations are replaced by the multiplication operations.

	\subsection{Function Secret Sharing}
	\label{sec:fss}
	
	%FSS \cite{boyle2015function} provides a low-interaction extension of additive secret sharing to compute complex functions. Specifically, the FSS-based approaches offer significant savings in online communication and round complexity compared to other alternative techniques, such as garbled circuits \cite{yao1982protocols} or additive secret sharing \cite{mohassel2017secureml}. Therefore, when computing parties involved are in different trust domains where communication is limited and expensive, FSS has excellent performance over other alternatives. 
	
	Function secret sharing (FSS) \cite{boyle2015function} allows to split a private function $f$ into succinct function keys such that every key itself does not reveal private information about $f$.
	Each key can be evaluated at a given point $x$, and combining the evaluation results will produce $f(x)$.
	A two-party FSS-based scheme is formally described as follows.
	\begin{definition}
		\label{def:fss}
		A two-party FSS scheme for computing a private function $f$ consists of two probabilistic polynomial time (PPT) algorithms:
		(1) $(k_{1},k_{2})\leftarrow\mathsf{Gen}(1^{\lambda}, f)$: Given the description of $f$ and a security parameter $\lambda$, output two succinct FSS keys $k_{1},k_{2}$, each for one party.
		(2) $\langle f(x)\rangle_{i} \leftarrow\mathsf{Eval}(k_{i},x)$: Given an FSS key $k_{i}$ and input $x$, output the share $\langle f(x)\rangle_{i}$.% of the function evaluation result $f(x)$.
	\end{definition} 
	\noindent The security guarantee of FSS is that an adversary learning only one of the keys $k_{1}$ and $k_{2}$ learns no private information about the target function $f$ and output $f(x)$. 
	
	\section{Problem Statement}
		\label{sec:problem_def}
	\subsection{System Architecture}

	\begin{figure}
		\centering
		\includegraphics[width=0.7\linewidth]{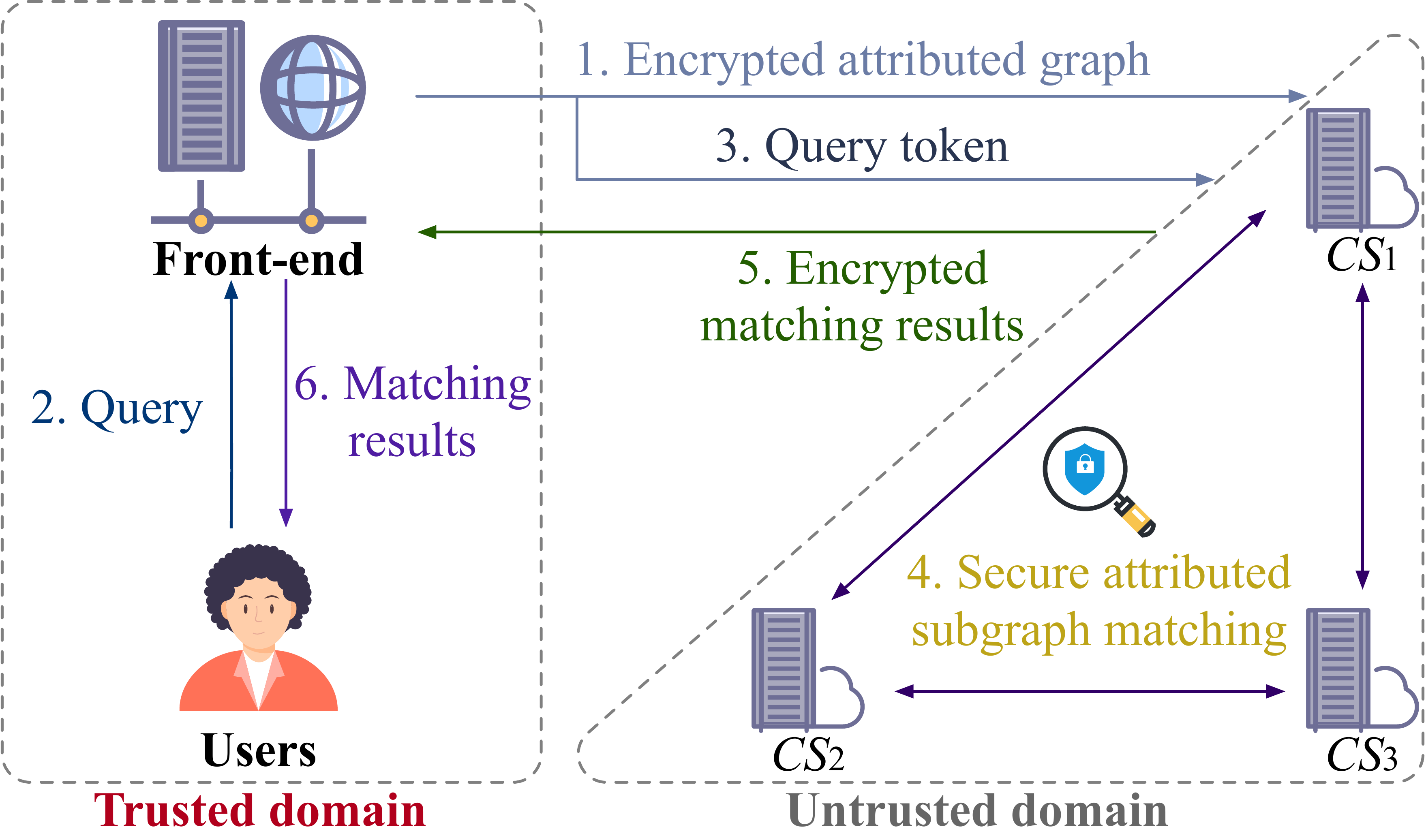}
		\caption{The system architecture of {\main}.}
		\label{fig:systemmodel}
		\vspace{-15pt}
	\end{figure}

	{\main} is aimed at supporting oblivious and encrypted subgraph matching services in cloud computing.
	Fig. \ref{fig:systemmodel} illustrates the system architecture of {\main}.
	There are three kinds of entities: the users, the graph owner as the on-premise service front-end (\cli), and the cloud servers.
	The graph owner (e.g., an enterprise or an organization) is in possession of large amounts of information modeled as an attributed graph, and wants to leverage the power of cloud computing for storing and querying this graph.
	We specifically consider that the graph owner expects cloud computing services to empower subgraph matching queries made by her users (e.g., an enterprise's employees or consumers) over the attributed graph, a popular and highly useful class of graph queries.
	Such cloud-empowered graph analytics service paradigm has seen wide adoption in practice (e.g., \cite{airbnb,PIXNET}, to list a few).
	However, due to privacy concerns on the proprietary attributed graph and queries, it is demanded that security must be embedded in such service \revise{paradigm from the very beginning}, providing protection for the outsourced attributed graph, subgraph matching queries, and the query results.

	In {\main}, the power of cloud is split into three cloud servers (referred to simply as {\csa} hereafter) from different trust domains, which can be hosted by independent cloud providers in practice. 
	Such multi-server model has gained rising popularity in recent years for building practical secure systems in both academia \cite{dauterman2020dory,chen2020metal,xu2021privacy,tan2021cryptgpu,dauterman2022waldo,wang2022privacy} and industry \cite{Mozilla,WhitePaper}. 
	%
	% It presents compatibility with lightweight secret sharing techniques.  
	%
	{\main} follows such trend and newly explores the support for oblivious and encrypted attributed subgraph matching services in cloud computing.

	\subsection{Threat Model and Security Guarantees}
	\label{sec:threat_model}

	\noindent\textbf{Threat model.} Similar to prior work using the multi-server setting for security designs \cite{chen2020metal,tan2021cryptgpu,ding2021efficient,hamlin2021two,zhou2019lppa,wang2022PeGraph}, we consider a semi-honest and non-colluding adversary model where each of {\csa} honestly follows our protocol, but may \textit{individually} attempt to infer the private information. 
	In addition, we assume that {\cli} and the users are trustworthy parties, since {\cli} is the owner of attributed graph, and can restrict the ranges of queries that different users are allowed to make using a standard database access control list \cite{de2003access}.
	
	%and {\cli} manages the access control for different users. Because {\cli} knows the identity of each user, the access control is straightforward, namely, {\cli} can restrict the ranges of queries that different users are allowed to make using a standard database access control list \cite{de2003access}.

	%
	\noindent\textbf{Security guarantees.} Under the aforementioned semi-honest and non-colluding adversary model, {\main} guarantees that the cloud servers cannot learn 1) each vertex's attribute values and exact degree, and the connections between these vertices, in the attributed graph; 2) the target attribute's value(s) associated with each vertex in subgraph queries; 3) search access patterns. 
	To define search access patterns for oblivious attributed subgraph matching, we adapt the general definitions in searchable encryption \cite{curtmola2006searchable}, which are introduced below. 
	% matching include \textit{search pattern} and \textit{access pattern}. 
	%
	% We give their formal definitions as follows.
	\begin{definition}
		\label{def:query_pattern}
		\textbf{Search pattern.} For two subgraph queries $q$ and $q'$, define $(q\overset{?}{=} q')\in\{0,1\}$, where if the two queries are identical, $(q\overset{?}{=} q')=1$, and otherwise $(q\overset{?}{=} q')=0$, and ``identical'' means that both the structure, vertices, and attribute values of $q$ and $q'$ are identical. %the structure, vertices' types,  attribute types of vertices as well as predicates associated
		Let $\mathbf{q}=\{q_{1},\cdots,q_{m}\}$ be a non-empty sequence of queries. 
		The search pattern reveals an $m\times m$ (symmetric) matrix with entry $(i,j)$ equals $(q_{i}\overset{?}{=} q_{j})$.
	\end{definition}

	\noindent In practice, the search pattern implies whether a new subgraph query has been issued before. 
	\begin{definition}
		\label{def:access_pattern}
		\textbf{Access pattern.} Given a subgraph query $q$ on the attributed graph $\mathcal{G}$, the access pattern reveals $\{g_{m}=(\mathcal{V}_{m}, \mathcal{E}_{m})\}$, where $g_{m}$ denotes a subgraph of $\mathcal{G}$ isomorphic to $q$.
	\end{definition}
	\noindent In practice, the access pattern reveals which vertices are ``accessed'', namely, which vertices in $\mathcal{G}$ are matched with the vertices in $q$. 
	In addition, the access pattern also implicitly reveals the connections between vertices in $\mathcal{G}$ because $g_{m}$ is isomorphic to $q$ and $q$'s structure is public.
	Notably, protecting the search access patterns can defend against a large class of potential leakage-abuse attacks \cite{kornaropoulos2020state,oya2021hiding,damie2021highly}.

	Similar to the prior works \cite{chang2016privacy,huang2021privacy}, {\main} considers the following information as public: 1) the schema layout parameters of the attributed graph and subgraph queries, including the number and types of vertices and edges and the types of vertex attributes; 2) the type of predicate associated with each vertex in queries, i.e., whether the predicate is an equality, single-sided range or interval range predicate; 3) the structure of queries. % and the corresponding matching results. In particular, the predicate type reveals whether the predicate is an equality or range predicate and the vertex attribute corresponding to the predicate. 
	To make the public information of queries more concrete, we consider the query $q$ in Fig. \ref{fig:viewgraph}: the attacker learns that the query is  
	\begin{equation}
		\label{eq:structure_q}
		\mathtt{U} ~(\mathtt{Place}=*)
		\begin{array}{cccc}
			\nearrow\mathtt{P}(\mathtt{Age}\in[*,*])\to\mathtt{C}(\mathtt{Field}=*)\\
			\searrow \mathtt{P}(\mathtt{Age}\in[*,*])\to\mathtt{C}(\mathtt{Field}=*)
		\end{array} .
	\end{equation}
	
	\section{The Design of {\main}}
		\label{sec:main}

	From a high-level point of view, {\main} proceeds through the following four phases: 1) attributed graph and subgraph queries modelling, 2) attributed graph encryption, 3) secure query token generation, and 4) secure attributed subgraph matching. 
	In phase 1, {\cli} properly models attributed graphs and subgraph queries so as to facilitate the subsequent oblivious subgraph matching service.
	In phase 2, {\cli} adequately encrypts its attributed graph and then sends the resulting ciphertext to the cloud servers. 
	In phase 3, {\cli} parses each subgraph query and generates the corresponding secure query token, followed by sending it to the cloud servers.
	In phase 4, the cloud servers obliviously retrieve encrypted subgraphs isomorphic to the query from the encrypted attributed graph.
	In what follows, we elaborate on the details of each phase.

	\subsection{Attributed Graph and Subgraph Queries Modelling}
	\label{sec:modelling}

	% We start with introducing how attributed graphs and subgraph queries are adequately modeled in {\main} so as to support the subsequent secure subgraph matching service.
	
	\noindent\textbf{Attributed graph modelling.} To represent the structure and non-structure information of an attributed graph $\mathcal{G}$, our main insight is to delicately adapt the inverted index structure \cite{curtiss2013unicorn}.
	Specifically, given a vertex $\mathtt{V}_{i}\in\mathcal{G}$, we first represent $\mathtt{V}_{i}$'s each attribute as a tuple $(t_{j}, d_{j}),j\in[S]$, where $S$ is the number of $\mathtt{V}_{i}$'s attributes (we write $[S]$ for the set $\{1,2,\cdots,S\}$), $t_{j}$ and $d_{j}$ are the type and value of the attribute, respectively.
	Then $\mathtt{V}_{i}$ can be modeled as $\mathtt{V}_{i}=\{T_{i},id_{i},\{(t_{j}, d_{j})\}_{j\in[S]}\}$, where $T_{i}$ is $\mathtt{V}_{i}$'s type and $id_{i}$ is $\mathtt{V}_{i}$'s identifier (ID) (i.e., a unique number).
	In this paper, for clarity of presentation, we use $\{\sigma_{i}\}_{i\in[\mu]}$ to represent the set $\{\sigma_{1},\cdots,\sigma_{\mu}\}$, and omit the subscript $i\in[\mu]$ when it does not affect the presentation.
	It is noted that the ID can be regarded as a special attribute with unique value for each vertex in $\mathcal{G}$. 
	%
	%In this paper, for clarify of presentation, we use ``$a.b$" to denote the relationship between $a, b$, e.g.,  $\mathtt{V}_{i}.T$ denotes $\mathtt{V}_{i}$'s type, $\mathtt{V}_{i}.id$ denotes $\mathtt{V}_{i}$'s ID,  and $\mathtt{at}_{j}.d$ denotes attribute $\mathtt{at}_{j}$'s value.
	%
	Then we consider how to model the connections between vertices. 
	Since the connection types in an attributed graph are varying, to clearly distinguish between different types, we associate each vertex $\mathtt{V}_{i}$ with several \textit{posting lists}, each containing the IDs of $\mathtt{V}_{i}$'s neighboring vertices with the same type.
	A posting list with $\mathtt{V}_{i}$ is represented as $P^{T_{ne}}_{\mathtt{V}_{i}}=\{id_{i,j}\}_{j\in[L]}$, where $id_{i,j},j\in[L]$ is the ID of $\mathtt{V}_{i}$'s each neighboring vertex with type $T_{ne}$ and $L$ is the number of them, i.e., $L=|P^{T_{ne}}_{\mathtt{V}_{i}}|$. 
	Therefore, the neighboring vertices of $\mathtt{V}_{i}$ can be represented as $\{ P^{T_{ne}}_{\mathtt{V}_{i}}\}_{T_{ne}\in\mathcal{TN}}$, where $\mathcal{TN}$ is a set of types for the posting lists of $\mathtt{V}_{i}$.

	%
	%It is noted that a vertex can have different connections with the vertices with the same type, e.g., the connections between person and person can be friend or follower. 
	%
	%In this case, we can put them in different posting lists and distinguish them by the type information $T_{ne}$.

	\noindent\textbf{Subgraph queries modelling.} We then consider how to properly model a subgraph query $q$. 
	Given a vertex $\mathtt{V}_{i}$ (named as \textit{target vertex}) in $q$, $\mathtt{V}_{i}$ has the \textit{target type} $T_{i}$ and \textit{target attribute} $(t_{i}, pd_{i})$, where $t_{i}$ is the type of the target attribute and $pd_{i}$ indicates the predicate associated with the target attribute. 
	It is noted that $pd_{i}$ can be an exact value, indicating an equality predicate, or a range (single-sided or an interval), indicating a range predicate, which are corresponding to exact matching and range matching, respectively. 
	For simplicity, we assume that each target vertex $\mathtt{V}_{i}$ only has one target attribute, but more attributes are straightforward, which will be introduced shortly in Section \ref{sec:match}. 
	Therefore, a query $q$ can be modeled as $q=\{\mathtt{V}_{i}=(T_{i},(t_{i}, pd_{i}))\}_{i\in[|q|]}$, where $|q|$ is the number of target vertices in $q$.
	In addition, the connections of vertices in $q$ can be simply represented as physical connections (e.g., via pointers).
	To make the query modelling more concrete, we consider the query $q$ in Fig. \ref{fig:viewgraph}, which can be modeled as 
	\begin{align}\notag
		q:=\{&(\mathsf{U}, (\mathsf{Place}, ``Harbin")),\\\notag
		&(\mathsf{P}, (\mathsf{Age},``[30,40]")),
		(\mathsf{P}, (\mathsf{Age},``[30,40]")),\\\notag
		&(\mathsf{C}, (\mathsf{Field}, ``software")),
		(\mathsf{C},(\mathsf{Field}, ``Internet"))
		\}, 
	\end{align}
	and its structure is same as Eq. \ref{eq:structure_q}.
	
	\subsection{Attributed Graph Encryption} 
	\label{sec:enc}
	
	We now introduce how an attributed graph is encrypted in {\main} so as to support the subsequent secure subgraph matching service.
	Here we need to encrypt for each vertex the values of its associated attributes and the IDs in its associated posting lists.
	To achieve high efficiency with lightweight secret sharing techniques, one plausible approach is to apply over each value the common 2-out-of-2 additive secret sharing technique \cite{demmler2015aby}.
	With such technique, a secret value $x\in  \mathbb{Z}_{2^{k}}$ is split into two shares $\langle x\rangle_{1}, \langle x\rangle_{2}\in  \mathbb{Z}_{2^{k}}$ such that $x =\langle x\rangle_{1}+ \langle x\rangle_{2}$ in $\mathbb{Z}_{2^{k}}$.
	%
	% Each share alone reveals no information about $x$.
	%
	However, to support multiplication over two secret-shared values, such technique requires one-round communication among the cloud servers holding the shares.
	In the multi-server model, it is highly desirable to make the communication among the cloud servers as little as possible.
	Therefore, instead of using the standard additive secret sharing technique, {\main} builds on the technique of RSS \cite{araki2016high} under the three-server model, which allows local operations for the cloud servers to perform a number of multiplications for secret-shared values and aggregate the results.

	However, {\main} does not directly use RSS to encrypt each value of the attributes and IDs in the posting lists.
	%
	% However, such approach will lead to high communication cost.
	%
	Instead, {\main} encodes each value of the attributes and IDs in the posting lists as a \textit{one-hot vector}, where all entries are ``0'' except for the entry at the location corresponding to the value which is set to ``1''.
	The RSS technique is then applied over these one-hot vectors.
	As will be clear later in Section \ref{sec:token_gen}, such encoding strategy will benefit the subsequent secure query process, allowing high communication efficiency in sending the secure query token in {\main}.

	With the above design intuition, we now describe how {\cli} encrypts the attributed graph for outsourcing.
	Note that the above attributed graph modeling allows {\cli} to simply perform encryption for each vertex separately.
	Specifically, given a vertex $\mathtt{V}_{i}\in \mathcal{G}$, {\cli} first encodes its each attribute value and each ID in the posting lists into a one-hot vector.
	To save the storage cost, we encode the IDs of vertices with different types separately, and thus the lengths of IDs of vertices with different types are varying.
	After that, {\cli} encrypts these one-hot vectors via RSS in the binary domain: 1) $ \mathtt{V}_{i}=\{T_{i},\llbracket   \mathbf{id}_{i}\rrbracket, \{(t_{j}, \llbracket \mathbf{d}_{j}\rrbracket)\}_{j\in[S]}\}$, where a one-hot vector is written in bold; 2) $\llbracket P^{T_{ne}}_{\mathtt{V}_{i}}\rrbracket=\{\llbracket \mathbf{id}_{i,j}\rrbracket\}_{j\in[L]}$ for each posting list with type $T_{ne}$.

	\revise{
	It is noted that {\cli} does not protect the type information (i.e., $T_{i}, \{t_{j}\}_{j\in[S]}, T_{ne}$) and the number of vertices with the same type (reflected by the length of vertex IDs), because they are insensitive public information \cite{chang2016privacy,huang2021privacy}.
	More specifically, the type of a vertex indicates only the public general category of the corresponding entity. For example, it indicates that the entity is a person, a university, or a company. Vertices with the same type must have the same attribute types, e.g., all persons have the attribute of ``age'' and all universities have the attribute of ``location''.
	The vertices with the same type $T_i$ have the same length of IDs, which indicates only the number of vertices having the public type $T_i$.
	Therefore, since the type information is common to different vertices, we consider it as public. 
	Such treatment also appears in prior works \cite{chang2016privacy,huang2021privacy}.
	}

	A remaining challenge here is that simply encrypting the IDs in each posting list without protecting the length information will leak the vertex degree, which could be exploited by inference attacks \cite{zhou2008brief}. 
	To tackle this challenge, {\main} adapts the idea of $k$-automorphism \cite{zou2009k} and has {\cli} blend some $0$-vectors as dummy IDs into each $\mathtt{V}_{i}$'s posting lists to make at least $k-1$ other vertices with same type as $\mathtt{V}_{i}$ have the same degree as $\mathtt{V}_{i}$. 
	More specifically, we note that each vertex with the same type has the same types of posting lists but their lengths can vary, e.g., each $\mathtt{Person}$ vertex has a friend list and a follower list, but the numbers of their friends and followers are varying.
	Therefore, given vertex $\mathtt{V}_{i}$ of type $T_{i}$ and posting lists $\{ P^{T_{ne}}_{\mathtt{V}_{i}}\}_{T_{ne}\in\mathcal{TN}}$, {\main} has {\cli} first find $k-1$ other vertices $\{\mathtt{V}_{p}\}_{p\in[k-1]}$ with type $T_{i}$, where each $\mathtt{V}_{p}$'s posting list with type $T_{ne}$ has the similar length as that of $\mathtt{V}_{i}$, i.e., $|P^{T_{ne}}_{\mathtt{V}_{1}}|\approx\cdots\approx |P^{T_{ne}}_{\mathtt{V}_{k-1}}|\approx |P^{T_{ne}}_{\mathtt{V}_{i}}|, T_{ne}\in\mathcal{TN}$.
	Then {\cli} blends some $0$-vectors as dummy IDs into them to achieve $|P^{T_{ne}}_{\mathtt{V}_{1}}|=\cdots= |P^{T_{ne}}_{\mathtt{V}_{k-1}}|= |P^{T_{ne}}_{\mathtt{V}_{i}}|, T_{ne}\in\mathcal{TN}$.
	After that, {\main} lets {\cli} apply RSS over the true and dummy IDs.
	%
	% Under RSS, even encrypting the same (zero) value multiple times will result in different secret shares (i.e., ciphertexts) indistinguishable from uniformly random values.
	%
	% Therefore, the dummy IDs are indistinguishable from the true IDs given that RSS is applied. 
	%
	% The effect of the dummy IDs will be eliminated in our subsequent design (Section \ref{sec:match}).
	%
	Since the attribute values are also encrypted in RSS, each vertex has at least $k-1$ other ``symmetric vertices'' in $\mathcal{G}$ and the encrypted attributed graph is a $k$-automorphism graph \cite{zou2009k}.

	Finally, the ciphertext of $\mathcal{G}$ can be represented as $\llbracket\mathcal{G}^{k}\rrbracket=\{( \mathtt{V}_{i},\{\llbracket P^{T_{ne}}_{\mathtt{V}_{i}}\rrbracket\}_{T_{ne}\in\mathcal{TN}})\}_{i\in[N]}$, where $\mathcal{TN}$ is a set of posting lists' types of $\mathtt{V}_{i}$ and $N$ is the number of vertices in $\mathcal{G}$. %, where $\llbracket \mathtt{V}_{i}\rrbracket:=\{T,(,\llbracket   \mathbf{id}\rrbracket), (t_{j}, \llbracket \mathbf{d}_{j}\rrbracket),j\in[S]\}$ and $\llbracket P^{T_{ne}}_{\mathtt{V}_{i}}\rrbracket=\{\llbracket \mathbf{id}_{j}\rrbracket,j\in[L]\}$.
	\revise{Algorithm \ref{alg:enc} describes how {\cli} encrypts $\mathcal{G}$.}
	{\cli} sends the public information and the RSS shares of $\llbracket\mathcal{G}^{k}\rrbracket$ to {\csa}, respectively.

		\begin{algorithm}[!t]
		\caption{\revise{Attributed Graph Encryption}}
		\label{alg:enc}
		\begin{algorithmic}[1] 

			\REQUIRE \revise{The attributed graph $\mathcal{G}$.}
			\ENSURE \revise{The encrypted attributed graph $\llbracket \mathcal{G}^{k}\rrbracket$.}
			\STATE \revise{Initialize an empty set $\llbracket \mathcal{G}^{k}\rrbracket^{B}=\emptyset$.}
			\WHILE{\revise{$\mathcal{G}\neq\emptyset$}}
			\STATE \revise{Select $k$ vertices $\{\mathtt{V}_{p}\}_{p\in[k]}$ with the same type from $\mathcal{G}$, and then delete $\{\mathtt{V}_{p}\}_{p\in[k]}$ from $\mathcal{G}$.}
			
			\revise{\# Protect the degrees of $\{\mathtt{V}_{p}\}_{p\in[k]}$:}
			\FOR{\revise{$T_{ne}\in\mathcal{TN}$}}
			\STATE \revise{Blend some dummy IDs into $\{P^{T_{ne}}_{\mathtt{V}_{p}}\}_{p\in[k]}$ to achieve $|P^{T_{ne}}_{\mathtt{V}_{1}}|=\cdots=|P^{T_{ne}}_{\mathtt{V}_{k}}|$.}
			\ENDFOR
			
			\revise{\# Encrypt the content of the padded vertices:
			\STATE Encode the private information of $\mathtt{V}_{p}, p\in[k]$ into one-hot vectors: $\mathtt{V}_{p}=\{T_p, \mathbf{id}_{p},  \{(t_{j}, \mathbf{d}_{j})\}_{j\in[S]}\}$ and $ P^{T_{ne}}_{\mathtt{V}_{p}}=\{ \mathbf{id}_{p,j}\}_{j\in[\hat{L}]},T_{ne}\in\mathcal{TN}$, where $\hat{L}$ is the length of positing lists after padding.
			\STATE Apply binary RSS over the one-hot vectors to produce the ciphertext $ \mathtt{V}_{p}=\{T_{p},\llbracket   \mathbf{id}_{p}\rrbracket, \{(t_{j}, \llbracket \mathbf{d}_{j}\rrbracket)\}_{j\in[S]}\}$ and $\llbracket P^{T_{ne}}_{\mathtt{V}_{p}}\rrbracket=\{\llbracket \mathbf{id}_{p,j}\rrbracket\}_{j\in[\hat{L}]},T_{ne}\in\mathcal{TN}$.
			\STATE $\llbracket \mathcal{G}^{k}\rrbracket^{B}.\mathsf{add}( \mathtt{V}_{p},\{\llbracket P^{T_{ne}}_{\mathtt{V}_{p}}\rrbracket\}_{T_{ne}\in\mathcal{TN}}),p\in[k]$.
		}
			\ENDWHILE
		
			\STATE \revise{Output the encrypted attributed graph $\llbracket \mathcal{G}^{k}\rrbracket$.}
		
		\end{algorithmic}
	\end{algorithm}

	\subsection{Secure Query Token Generation}
	\label{sec:token_gen}
	
	Given a subgraph query $q$, {\cli} then generates a secure query token in a custom way.
	As modeled in Section \ref{sec:modelling}, a subgraph query $q$ is in the form $q=\{\mathtt{V}_{i}=(T_{i},(t_{i}, pd_{i}))\}_{i\in[|q|]}$.
	What should be protected is the value information $pd_{i}$ for the predicate because $T_{i}$ and $t_{i}$ refer to the types of vertices and attributes, which are public information \cite{chang2016privacy,huang2021privacy}. 
	%
	% modeled as in Section \ref{sec:modelling}, the private information {\main} needs to protect is the predicates $\{pd_{i}\}$ because $\{T_{i}\},\{t_{i}\}$ are types of vertices and attributes and not sensitive information \cite{chang2016privacy,huang2021privacy}. 
	% \newpage
	% We now elaborate on that given a subgraph query $q$, how {\cli} parses it into a query token for the use at {\csa} side while without leaking the query pattern.
	% %
	% Given a query $q:=\{\mathtt{V}_{i}:=\{T_{i},(t_{i}, pd_{i})\}\}_{i\in[|q|]}$ modeled as in Section \ref{sec:modelling}, the private information {\main} needs to protect is the predicates $\{pd_{i}\}$ because $\{T_{i}\},\{t_{i}\}$ are types of vertices and attributes and not sensitive information \cite{chang2016privacy,huang2021privacy}. 
	%
	{\main} flexibly supports both equality predicate and range predicate, so $pd_i$ can refer to an exact value or a range.
	%
	% It is noted that the predicate values $pd_{i}$ can be classified into ``\textit{equality}'' (i.e., $pd_{i}$ is an exact value) and ``\textit{range}'' (i.e., $pd_{i}$ is a range) corresponding to exact matching and range marching, respectively. 
	%
	% In addition, {\main} provides support for that the predicate values in a query are equality or range, i.e., mixed matching.

	% We now consider how to generate the secure query token so that it can be efficiently and effectively evaluated on the cloud side.
	%
	With the secure subgraph matching service run among the three cloud servers, {\main} aims to minimize the communication among the cloud servers.
	We identify the newly developed technique---function secret sharing (FSS)---as an excellent fit for our purpose, which allows low-interaction secure evaluation of a function among multiple parties \cite{boyle2016function}. 
	Specifically, we observe two FSS constructions as a natural fit for the two types of predicates targeted in {\main}: distributed point functions (DPFs) \cite{boyle2016function} for equality predicates and distributed comparison functions (DCFs) \cite{boyle2021function} for range predicates.  
	The FSS-based DPF consists of the same algorithms as Definition \ref{def:fss}, which allows two servers to obliviously evaluate a point function $f^{=}_{\alpha, \beta}$, outputting secret-shared $\beta$ if input $\alpha$, otherwise, outputting secret-shared 0.
	Similarly, DCF is for a comparison function $g^{<}_{\alpha,\beta}$, which outputs secret-shared $\beta$ if $x<\alpha$, otherwise, outputs secret-shared 0.
	Analogously, DCF can also describe the functions $x>\alpha$, $x\le\alpha$ and $x\ge\alpha$.
	In addition, constructions for interval containment (IC) build on DCFs to express functions of the form $\alpha<x<\alpha'$ (denoted as $g^{\ll}_{\alpha,\alpha',\beta}$). 
	Analogously, IC can also describe the functions $\alpha\le x<\alpha'$, $\alpha< x\le\alpha'$ and $\alpha\le x\le\alpha'$.

	Applying the advanced FSS techniques in {\main}, however, is not straightforward and needs delicate treatment.
	In particular, in {\main} the values to be evaluated via the FSS technique are not in plaintext domain and each cloud server holds shares of the values. 
	However, the FSS-based evaluation process requires the cloud servers to work on identical inputs for producing correct outputs.
	To address this issue, one relatively simple yet effective approach as proposed by Boyle \textit{et al.} \cite{boyle2016function,boyle2021function} is to have the cloud servers open additively masked versions of secret values and tailor the generation of the FSS keys for evaluation over the masked values.
	While this basic approach can protect the secret value while allowing FSS-based evaluation, it has two critical limitations: i) evaluating the same private predicate on different encrypted attribute values in $\llbracket\mathcal{G}^{k}\rrbracket$ requires (a large number of) fresh FSS keys, imposing high computation and communication overhead on {\cli}; ii) the evaluation for each secret-shared value requires the cloud servers to have one-round communication (for opening a masked version of the secret value), leading to high cloud-side communication overhead either.

	As such, {\main} does not build on the above basic approach and makes a delicate treatment for high efficiency.
	It is recalled that in the attributed graph encryption phase, each value demanding protection is encoded into a one-hot vector.
	The adoption of such encoding strategy, inspired by \cite{dauterman2022waldo,dauterman2020dory}, is actually useful in providing an alternative way to avoid fresh FSS keys in {\main} for evaluating the same predicate on different attribute values.
	At a high level, the idea is that with such encoding, FSS keys can be evaluated against the public locations of entries in one-hot vectors.
	Then we can multiply the evaluation result on each entry with its value, and aggregate all multiplication results to produce the target evaluation result for a one-hot vector/secret value.
	On another hand, as the one-hot vectors are protected under the RSS technique in {\main}, the above idea has to be instantiated over the RSS-protected one-hot vectors.
	Inspired by \cite{dauterman2022waldo}, {\main} leverages the replication property of RSS and constructs three pairs of FSS keys for each private predicate so as to bridge FSS and RSS.
	How this can actually work out will become more clear in the subsequent phase of secure attributed subgraph matching, which will be introduced shortly in Section \ref{sec:match}.
	
	\begin{algorithm}[!t]
		\caption{Secure Attributed Subgraph Matching $\mathsf{secMatch}$} 
		\label{alg:main}
		\begin{algorithmic}[1] 
			\REQUIRE The encrypted attributed graph $\llbracket\mathcal{G}^{k}\rrbracket$; a secure query token $\mathsf{tok}_{q}=\{\mathtt{V}_{i}=(T_{i},(t_{i}, \mathcal{K}_{i}))\}_{i\in[|q|]}$.
			
			\ENSURE The encrypted matching results $\{\llbracket g_{m}\rrbracket\}$.
			
			\STATE \textbf{Initialization:} {\csa} initialize an empty set $\llbracket\mathcal{Q}\rrbracket$.
			
			\FOR{$i\in[|q|]$}
			\STATE $\mathtt{V}_{i}:=(T_{i},(t_{i}, \mathcal{K}_{i}))$.
			\revise{
			\IF{$\mathtt{V}_{i}$ is a start vertex in $\mathsf{tok}_{q}$}
				\STATE
			Set $\{\mathtt{V}_{c}\}$ as all vertices with type $T_{i}$ in $\llbracket\mathcal{G}^{k}\rrbracket$ and set $\{\llbracket \mathbf{id}_{\mathtt{V}_{c}}\rrbracket\} $ and $\{\llbracket \mathbf{d}_{\mathtt{V}_{c}}\rrbracket\} $ as these vertices' IDs and values of attribute with type $t_{i}$, respectively.
			\ENDIF
		}

			\STATE $\{\llbracket  x_{\mathtt{V}_{c}}\rrbracket\} =\mathsf{secEval}(\{\llbracket \mathbf{d}_{\mathtt{V}_{c}}\rrbracket\} ,\mathcal{K}_{i})$.
			
			\STATE  $(\{\llbracket \mathbf{id}_{\mathtt{V}_{m}}\rrbracket\} , \{\llbracket \mathbf{d}_{\mathtt{V}_{m}}\rrbracket\} )=\mathsf{secFetch}(\{\llbracket  x_{\mathtt{V}_{c}}\rrbracket\}, \{\llbracket \mathbf{id}_{\mathtt{V}_{c}}\rrbracket\},$ $\{\llbracket \mathbf{d}_{\mathtt{V}_{c}}\rrbracket\} )$.
			
			\STATE {\csa} add $(\{\llbracket \mathbf{id}_{\mathtt{V}_{m}}\rrbracket\} , \{\llbracket \mathbf{d}_{\mathtt{V}_{m}}\rrbracket\} )$ to $\llbracket\mathcal{Q}\rrbracket$.

			\STATE $(\{\llbracket \mathbf{id}_{\mathtt{V}_{ne}}\rrbracket\} ,\{\llbracket \mathbf{d}_{\mathtt{V}_{ne}}\rrbracket\} )=\mathsf{secAccess}(\{\llbracket \mathbf{id}_{\mathtt{V}_{m}}\rrbracket\},$ $ T_{ne},t_{ne})$. \# $T_{ne},t_{ne}$ is the type and target attribute type of $\mathtt{V}_{i}$'s one neighboring vertex in $\mathsf{tok}_{q}$, respectively.
			\STATE $\{\llbracket \mathbf{id}_{\mathtt{V}_{c}}\rrbracket\} =\{\llbracket \mathbf{id}_{\mathtt{V}_{ne}}\rrbracket\}$;  $\{\llbracket \mathbf{d}_{\mathtt{V}_{c}}\rrbracket\} =\{\llbracket \mathbf{d}_{\mathtt{V}_{ne}}\rrbracket\} $. 
			\ENDFOR
			\STATE {\csa} reorganize $\llbracket\mathcal{Q}\rrbracket$ into subgraphs $\{\llbracket g_{m}\rrbracket\}$.% based on the query's structure.% and delete incomplete subgraphs, and then outpu
		\end{algorithmic}
	\end{algorithm}
	
	With the above insights, we now introduce how to generate the secure query token.
	Specifically, given the secret value(s) $pd_{i}$ for a predicate, {\cli} generates three pairs of independent FSS keys $\mathcal{K}_{i}=\{(k_{1}^{1},k_{2}^{1}),(k_{1}^{2},k_{2}^{2}), (k_{1}^{3},k_{2}^{3})\}$ by setting $\alpha=pd_{i} $ and $\beta=1\in\mathbb{Z}_{2}$. 
	The output domain is set as $\mathbb{Z}_{2}$ to make it compatible with $\llbracket\mathcal{G}^{k}\rrbracket$. 
	By this way, {\cli} can encrypt a subgraph query $q=\{\mathtt{V}_{i}=(T_{i},(t_{i}, pd_{i}))\}_{i\in[|q|]}$ into the corresponding secure query token $\mathsf{tok}_{q}=\{\mathtt{V}_{i}=(T_{i},(t_{i}, \mathcal{K}_{i}))\}_{i\in[|q|]}$. 
	Finally, {\cli} sends $\mathsf{tok}^{(1)}_{q}=\{(T_{i},(t_{i}, (k_{1}^{1},k_{1}^{2})))\}$, $\mathsf{tok}^{(2)}_{q}=\{(T_{i},(t_{i}, (k_{2}^{2},k_{1}^{3})))\}$ and $\mathsf{tok}^{(3)}_{q}=\{(T_{i},(t_{i}, (k_{2}^{3},k_{2}^{1})))\}$ along with the structure of $q$ to {\cs}, {\css}, and {\csss}, respectively. 
	To make this more concrete, we consider the query in Fig. \ref{fig:viewgraph}, whose secure query token is in the form of 
	\begin{equation}
		\label{eq:enc_subg}
		\mathtt{U} ~(\mathtt{Place}=\mathcal{K}_{1})
		\begin{array}{cccc}
			\nearrow\mathtt{P} (\mathtt{Age}\in\mathcal{K}_{2})\to\mathtt{C}(\mathtt{Field}=\mathcal{K}_{4})\\
			\searrow \mathtt{P} (\mathtt{Age}\in\mathcal{K}_{3})\to\mathtt{C}(\mathtt{Field}=\mathcal{K}_{5})
		\end{array}, 
	\end{equation}
	where $\mathcal{K}_{1},\cdots, \mathcal{K}_{5}$ are generated with the same target function output $\beta=1\in\mathbb{Z}_{2}$, and correspond to $f^{=}_{\alpha=``Harbin"}$, $g^{\ll}_{\alpha=30,\alpha'=40}$, $g^{\ll}_{\alpha=30,\alpha'=40}$, $f^{=}_{\alpha=``Software"}$ and $f^{=}_{\alpha=``Internet"}$, respectively.

	\subsection{Secure Attributed Subgraph Matching}
	\label{sec:match}

	\noindent \textbf{Overview.} Upon receiving the secure query token $\mathsf{tok}_{q}$ from {\cli}, the cloud servers collaboratively perform the secure subgraph matching process over the encrypted attributed graph $\llbracket\mathcal{G}^{k}\rrbracket$ and obtain encrypted subgraphs $ \{\llbracket g_{m}\rrbracket\}$ that are isomorphic to $q$.
	{\main} provides techniques that allow the cloud servers to search over the encrypted attributed graph while being oblivious to search access patterns.
	Our construction is comprised of three components: secure predicate evaluation over candidate vertices (denoted as $\mathsf{secEval}$), secure matched vertices fetching (denoted as $\mathsf{secFetch}$), and secure neighboring vertices accessing (denoted as $\mathsf{secAccess}$).

	At a high level, secure subgraph matching proceeds as follows at the cloud in {\main}.
	Given a current target vertex $ \mathtt{V}_{i}\in\mathsf{tok}_{q}$, {\main} provides $\mathsf{secEval}$ to have the cloud servers first perform secure predicate evaluation over candidate vertices in the encrypted attributed graph (i.e., vertices with the same type as $\mathtt{V}_{i}$) and produce encrypted predicate evaluation results. 
	Then, based on the encrypted evaluation results, {\main} provides $\mathsf{secFetch}$ to allow the cloud servers to obliviously fetch the encrypted matched vertices which satisfy the predicates based on the encrypted predicate evaluation results produced from $\mathsf{secEval}$.
	Afterwards, based on each matched vertex's encrypted ID, {\main} then needs to allow the cloud servers to obliviously access the IDs and attribute values of each matched vertex's neighboring vertices, which are used as the candidate vertices of the next-hop target vertex in $\mathsf{tok}_{q}$. 
	The above process runs iteratively until all target vertices in $\mathsf{tok}_{q}$ are processed.
	Finally, {\csa} reorganize the matched vertices into subgraphs based on the public structure of $\mathsf{tok}_{q}$, and delete incomplete subgraphs who do not have the complete structure as $\mathsf{tok}_{q}$, and then output the final encrypted matching results $\{\llbracket g_{m}\rrbracket\}$.
	In Algorithm \ref{alg:main}, we give {\main}'s complete construction for the secure subgraph matching process at the cloud, which relies on the coordination of the three components: $\mathsf{secEval}$, $\mathsf{secFetch}$, and $\mathsf{secAccess}$, following the aforementioned workflow.
	\revise{For clarity, we illustrate the secure attributed subgraph matching process in Fig. \ref{fig:viewsecmatch}.}
	In what follows, we elaborate on the design of each component.

		\begin{figure}
		\centering
		\includegraphics[width=\linewidth]{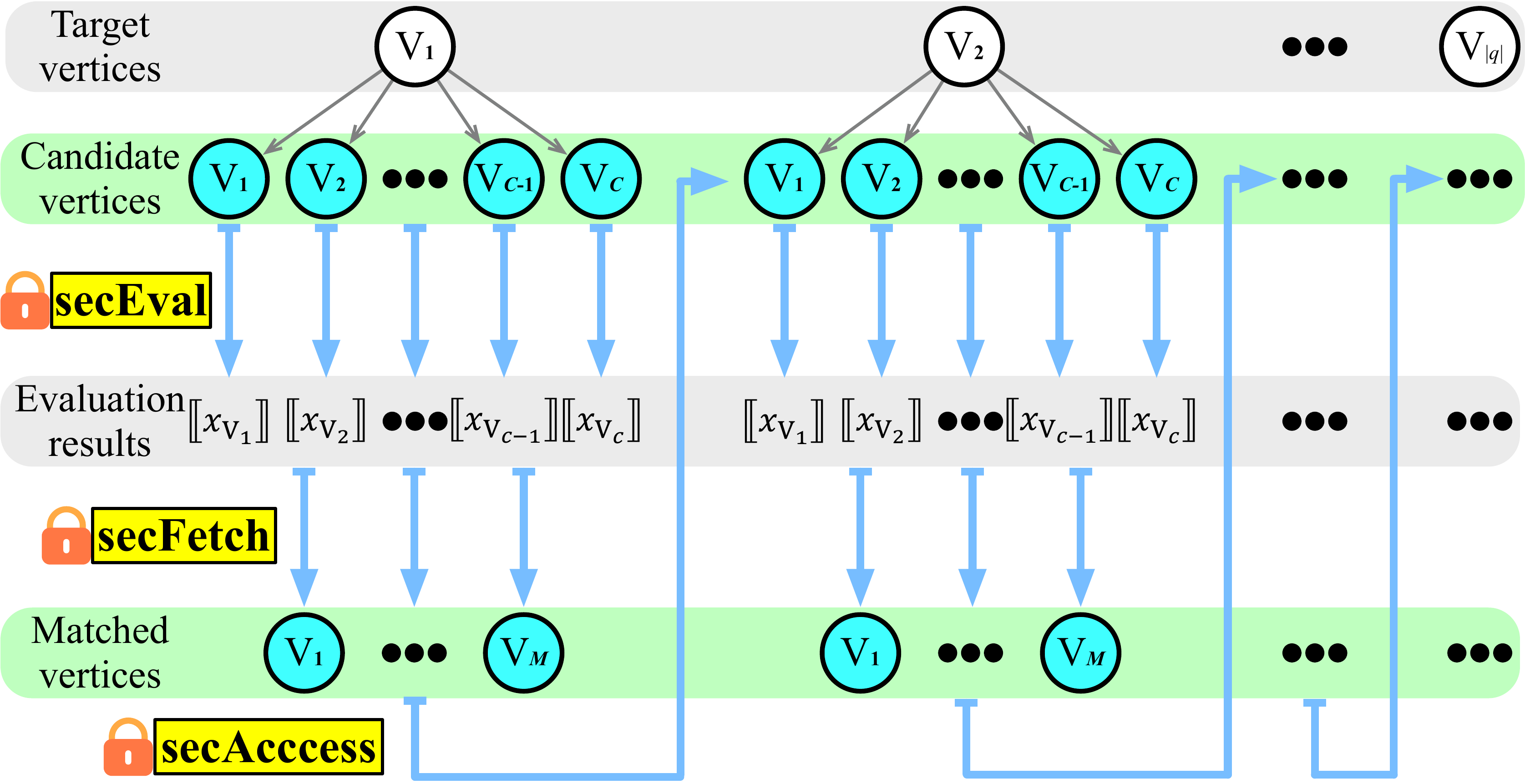}
		\caption{\revise{Illustration of secure attributed subgraph matching.}}
		\label{fig:viewsecmatch}
	\end{figure}

	\noindent\textbf{Secure predicate evaluation over candidate vertices.} 
	For simplicity of presentation, we start with introducing how to allow {\csa} to evaluate a single predicate over the candidate vertices for a target vertex in the secure query token. 
	%
	% The support for multiple predicates is straightforward and will be described later on.

	% Recall the subgraph queries modelling in Section \ref{sec:modelling}, we  assume that each target vertex in a query only has one target attribute (i.e., one predicate), and thus we next first introduce how {\csa} obliviously evaluate a single predicate over candidate vertices. 
	%
	% After that, we will introduce how {\csa} straightforwardly handle with the case of multiple predicates.

	Given a target vertex $\mathtt{V}_{i}=(T_{i},(t_{i}, \mathcal{K}_{i}) )\in\mathsf{tok}_{q}$, {\csa} need to first retrieve its candidate vertices $\{\mathtt{V}_{c}\}$'s $\{\llbracket \mathbf{id}_{\mathtt{V}_{c}}\rrbracket\} $ (i.e., IDs) and $\{\llbracket \mathbf{d}_{\mathtt{V}_{c}}\rrbracket\} $ (i.e., the values of the attribute with type $t_{i}$) from $\llbracket\mathcal{G}^{k}\rrbracket$. %, where candidate vertices $\{\mathtt{V}_{c}\}$ have the same type $T_{i}$ as $\mathtt{V}_{i}$.
	We note that there are two cases here that need to be treated separately:
	1) If $\mathtt{V}_{i}$ is a start vertex in the query and has no antecedent vertices (e.g., vertex $\mathsf{U}$ of query $q$ in Fig. \ref{fig:viewgraph}), then $\{\mathtt{V}_{c}\}$ are the vertices with type $T_{i}$ in $\llbracket\mathcal{G}^{k}\rrbracket$. {\csa} can locally set $\{\mathtt{V}_{c}\}$' IDs and the values of attribute  with type $t_{i}$ as $\{\llbracket \mathbf{id}_{\mathtt{V}_{c}}\rrbracket\} $ and $\{\llbracket \mathbf{d}_{\mathtt{V}_{c}}\rrbracket\} $, respectively; 
	2) Otherwise (i.e., $\mathtt{V}_{i}$ has an antecedent vertex), $\{\mathtt{V}_{c}\}$ are the neighboring vertices of $\mathtt{V}_{i}$'s antecedent vertex's matched vertices, and $\{\llbracket \mathbf{id}_{\mathtt{V}_{c}}\rrbracket\} $ and $\{\llbracket \mathbf{d}_{\mathtt{V}_{c}}\rrbracket\} $ will be obliviously retrieved through the component $\mathsf{secAccess}$, which will be introduced later.

	Then for each candidate vertex $\mathtt{V}_{c}$, {\csa} need to obliviously evaluate whether its attribute value $\llbracket \mathbf{d}_{\mathtt{V}_{c}}\rrbracket$ satisfies the encrypted predicate $\mathcal{K}_i$ associated with the target vertex $\mathtt{V}_{i}$. 
	Recall that each attribute value in {\main} is encoded as a one-hot vector and protected via RSS; and the encrypted predicate $\mathcal{K}_{i}$ consists of three pairs of FSS keys $\{(k_{1}^{1},k_{2}^{1}),(k_{1}^{2},k_{2}^{2}), (k_{1}^{3},k_{2}^{3})\}$.
	As shown in Algorithm \ref{alg:select}, the secure predicate evaluation in {\main} works as follows.
	%
	% The challenge here is how {\csa} obliviously evaluate $ \mathcal{K}_{i} $ on each encrypted $\llbracket \mathbf{d}_{\mathtt{V}_{c}}\rrbracket\in\{\llbracket \mathbf{d}_{\mathtt{V}_{c}}\rrbracket\}$ without requiring fresh FSS keys.% for each evaluation. 
	%
	% Our key insight is to let {\csa} obliviously evaluate the replicated FSS keys $ \mathcal{K}_{i} $ on RSS-protected one-hot vector $\llbracket \mathbf{d}_{\mathtt{V}_{c}}\rrbracket$ inspired by \cite{dauterman2022waldo}.
	% %
	% %Specifially, recall Section \ref{sec:enc}, we encode each attribute value $d$ into one-hot vector and encrypt them by binary RSS. 
	%
	For each secret-shared bit $\llbracket \mathbf{d}_{\mathtt{V}_{c}}[l]\rrbracket, l\in[n]$ ($n$ is the length of the one-hot vector), each of {\csa} first \textit{locally} evaluates the FSS keys it holds on the public location $l$, and then \textit{locally} ANDs the output by the private bit $\llbracket \mathbf{d}_{\mathtt{V}_{c}}[l]\rrbracket$.
	Then, each of {\csa} \textit{locally} XORs all results of AND operations to produce the encrypted predicate evaluation result $\llbracket x_{\mathtt{V}_{c}} \rrbracket$ (a secret-shared bit) about the candidate vertex $\mathtt{V}_{c}$.
	The secure predicate evaluation can be formally described as
	\begin{align}
		\label{eq:predEval}\notag
		\mathcal{CS}_{1}:
		\begin{array}{rl}
			&x^{(1)}_{1}=\bigoplus_{l=1}^{n}\mathsf{Eval}(k_{1}^{1},l)\otimes\langle \mathbf{d}_{\mathtt{V}_{c}}[l]\rangle_{1};\\\notag
			&x^{(2)}_{1}=\bigoplus_{l=1}^{n}\mathsf{Eval}(k_{1}^{2},l)\otimes\langle \mathbf{d}_{\mathtt{V}_{c}}[l]\rangle_{2},
		\end{array}\\
		\mathcal{CS}_{2}: 	
		\begin{array}{rl}
			&x^{(2)}_{2}=\bigoplus_{l=1}^{n}\mathsf{Eval}(k_{2}^{2},l)\otimes\langle \mathbf{d}_{\mathtt{V}_{c}}[l]\rangle_{2};\\ &x^{(3)}_{1}=\bigoplus_{l=1}^{n}\mathsf{Eval}(k_{1}^{3},l)\otimes\langle \mathbf{d}_{\mathtt{V}_{c}}[l]\rangle_{3},
		\end{array}\\
		\mathcal{CS}_{3}: 	
		\begin{array}{rl}
			&x^{(3)}_{2}=\bigoplus_{l=1}^{n}\mathsf{Eval}(k_{2}^{3},l)\otimes\langle \mathbf{d}_{\mathtt{V}_{c}}[l]\rangle_{3};\\
			&x^{(1)}_{2}=\bigoplus_{l=1}^{n}\mathsf{Eval}(k_{2}^{1},l)\otimes\langle \mathbf{d}_{\mathtt{V}_{c}}[l]\rangle_{1}.\notag
		\end{array}
	\end{align}

	%
	% Algorithm \ref{alg:select} provides the construction of the $\mathsf{SecEval}$ component.
	%
	It is noted that $ x_{\mathtt{V}_{c}}=x^{(1)}_{1}\oplus x^{(2)}_{1}\oplus x^{(2)}_{2}\oplus x^{(3)}_{1}\oplus x^{(3)}_{2}\oplus x^{(1)}_{2}$, which indicates whether $\mathbf{d}_{\mathtt{V}_{c}}$ satisfies the encrypted predicate $ \mathcal{K}_{i} $.
	Namely, $ x_{\mathtt{V}_{c}}=1$ indicates that $\mathtt{V}_{c}$ is a matched vertex and $x_{\mathtt{V}_{c}}=0$ indicates not.
	Note that through the secure evaluation via Eq. \ref{eq:predEval}, the result $x_{\mathtt{V}_{c}}$ is additively secret-shared.
	To be compatible with subsequent processing over RSS-based secret shares, the re-sharing operation as introduced in Section \ref{sec:rss} is performed so as to make the $x_{\mathtt{V}_{c}}$ in the RSS form.

	In the above for simplicity of presentation we focus on the case that a target vertex is associated with a single predicate.
	For the case where a target vertex in $\mathsf{tok}_{q}$ has multiple predicates, say $p$ predicates, {\csa} can first evaluate each predicate separately, outputting different results $\llbracket x^{1}_{\mathtt{V}_{c}} \rrbracket,\cdots,\llbracket x^{p}_{\mathtt{V}_{c}} \rrbracket$.
	Then {\csa} can flexibly aggregate them based on the Boolean expression specified by {\cli}.
	For example, if {\cli} requires $\mathtt{V}_{c}$ to satisfy all predicates, {\csa} can obliviously aggregate them by $\llbracket x_{\mathtt{V}_{c}} \rrbracket=\llbracket x^{1}_{\mathtt{V}_{c}} \rrbracket\otimes\cdots\otimes\llbracket x^{p}_{\mathtt{V}_{c}} \rrbracket$; if {\cli} only requires $\mathtt{V}_{c}$  to satisfy one of the predicates, {\csa} can obliviously aggregate them by  $\llbracket x_{\mathtt{V}_{c}} \rrbracket=\llbracket x^{1}_{\mathtt{V}_{c}} \rrbracket\oplus\cdots\oplus\llbracket x^{p}_{\mathtt{V}_{c}} \rrbracket$.

	\begin{algorithm}[!t]
		\caption{Secure Predicate Evaluation $\mathsf{secEval}$}
		\label{alg:select}
		\begin{algorithmic}[1] 
			\REQUIRE The candidate vertices' attribute values $\{\llbracket \mathbf{d}_{\mathtt{V}_{c}}\rrbracket\} $ and the encrypted predicate $\mathcal{K}_{i}$.
			
			\ENSURE  The encrypted evaluation results $\{\llbracket  x_{\mathtt{V}_{c}}\rrbracket\} $.
			
			\FOR{$\llbracket \mathbf{d}_{\mathtt{V}_{c}}\rrbracket\in \{\llbracket \mathbf{d}_{\mathtt{V}_{c}}\rrbracket\} $}
			
			\STATE  {\csa} locally evaluate $\mathcal{K}_{i}$ on $\{\llbracket \mathbf{d}_{\mathtt{V}_{c}}\rrbracket\} $ by Eq. \ref{eq:predEval}.
			\STATE {\csa} re-share the results to achieve $\llbracket x_{\mathtt{V}_{c}}\rrbracket$ in RSS.
			\ENDFOR
			
			\STATE {\csa} output the encrypted results $\{\llbracket  x_{\mathtt{V}_{c}}\rrbracket\} $.
		\end{algorithmic}
	\end{algorithm}

	\noindent\textbf{Secure matched vertices fetching.} With the encrypted predicate evaluation result $\llbracket x_{\mathtt{V}_{c}}\rrbracket$ produced for each candidate vertex $\mathtt{V}_{c}$, {\csa} then need to fetch information of the matched vertices that have $x_{\mathtt{V}_{c}}=1$.
	We denote the matched vertices by $\{\mathtt{V}_{m}\}$.
	Simply opening the evaluation result $x_{\mathtt{V}_{c}}$ for each candidate vertex to identify the matched vertices will easily violate the security requirement for access pattern protection.
	Instead, {\main} devises a component $\mathsf{secFetch}$, as given in Algorithm \ref{alg:reveal}, to allow {\csa} to \emph{obliviously} fetch the information about the matched vertices.
	Specifically, {\csa} should be able to fetch $\{\mathtt{V}_{m}\}$'s IDs $\{\llbracket \mathbf{id}_{\mathtt{V}_{m}} \rrbracket\}$ and attribute values $\{\llbracket \mathbf{d}_{\mathtt{V}_{m}} \rrbracket\}$, without knowing which candidate vertices are the matched ones.
	It is noted that there are two cases here to be treated separately: (1) \textit{Only one} candidate vertex is the matched vertex. This case corresponds to a target attribute for which each vertex has a unique value, e.g., when the target attribute  is ``$\mathtt{ID}$'' or ``$\mathtt{Phone ~number}$''. (2) \textit{Two or more} candidate vertices are the matched vertices. This corresponds to a target attribute for which each vertex does not have a unique value, e.g., when the target attribute is ``$\mathtt{Age}$''.

	\iffalse
	\begin{itemize}
		\item \textit{Case I}: \textit{Only one} candidate vertex is the matched vertex. This case corresponds to a target attribute for which each vertex has a unique value, e.g., when the target attribute  is ``$\mathtt{ID}$'' or ``$\mathtt{Phone ~number}$''.

		%\textit{only one} vertex's attribute value satisfies the predicate $\mathcal{K}_{i} $, e.g., the attribute is ``$\mathsf{Identifier}$'' or ``$\mathtt{Phone ~number}$'', because they are unique. 

		\item \textit{Case II}: \textit{Two or more} candidate vertices are the matched vertices. This corresponds to a target attribute for which each vertex does not have a unique value, e.g., when the target attribute is ``$\mathtt{Age}$''.

		%vertices' attribute values satisfy the predicate $\mathcal{K}_{i} $, e.g., the attribute is ``$\mathsf{Location}$'', because many entities have the same location.
	\end{itemize}
	
	\fi
	
	The above two cases can be distinguished by the cloud servers since the type information is public.
	We now introduce how {\main} deals with each case respectively.
	%
	% It is noted that since the attribute types are public, {\csa} can simply distinguish between the two cases. 
	%
	The first case can be easily handled as follows.
	{\main} lets {\csa} obliviously AND each candidate vertex $\mathtt{V}_{c}$'s $\llbracket \mathbf{id}_{\mathtt{V}_{c}}\rrbracket$ and $\llbracket \mathbf{d}_{\mathtt{V}_{c}} \rrbracket$ by its encrypted predicate evaluation result $\llbracket x_{\mathtt{V}_{c}} \rrbracket$ respectively, and then XOR the AND operation results to obtain the only matched vertex $\mathtt{V}_{m}$'s $\llbracket \mathbf{id}_{\mathtt{V}_{m}}\rrbracket$/$\llbracket \mathbf{d}_{\mathtt{V}_{m}}\rrbracket$ respectively. 
	Formally, {\csa} perform the following:
	\begin{align}\notag
		\begin{array}{rl}
			\llbracket \mathbf{id}_{\mathtt{V}_{m}}\rrbracket=&\bigoplus_{c=1}^{C}\llbracket \mathbf{id}_{\mathtt{V}_{c}}\rrbracket\otimes\llbracket x_{\mathtt{V}_{c}} \rrbracket;\\\notag
			\llbracket \mathbf{d}_{\mathtt{V}_{m}}\rrbracket=&\bigoplus_{c=1}^{C}\llbracket \mathbf{d}_{\mathtt{V}_{c}}\rrbracket\otimes\llbracket x_{\mathtt{V}_{c}} \rrbracket,
		\end{array}
	\end{align}
	where $C$ is the number of candidate vertices. 
	Correctness holds since only one candidate vertex is with $\llbracket x_{\mathtt{V}_{c}} \rrbracket=\llbracket 1 \rrbracket$, while others have $\llbracket x_{\mathtt{V}_{c}} \rrbracket=\llbracket 0 \rrbracket$.

	\begin{algorithm}[!t]
		\caption{Secure Matched Vertices Fetching $\mathsf{secFetch}$}%(\{\llbracket x_{\mathtt{V}_{c}}\rrbracket\} , \{\llbracket \mathbf{id}_{\mathtt{V}_{c}}\rrbracket\} ,\{\llbracket \mathbf{d}_{\mathtt{V}_{c}}\rrbracket\} )$} 
	\label{alg:reveal}
	\begin{algorithmic}[1] 
		\REQUIRE $\{\mathtt{V}_{c}\}$'s $\{\llbracket x_{\mathtt{V}_{c}}\rrbracket\} , \{\llbracket \mathbf{id}_{\mathtt{V}_{c}}\rrbracket\} $ and $\{\llbracket \mathbf{d}_{\mathtt{V}_{c}}\rrbracket\} $.
		
		%\ENSURE The matched vertices' encrypted IDs $\{\llbracket \mathbf{id}_{\mathtt{V}_{m}}\rrbracket\} $ and attribute values $\{\llbracket \mathbf{d}_{\mathtt{V}_{m}}\rrbracket\} $.
		
		\ENSURE Matched vertices $\{\mathtt{V}_{m}\}$'s $\{\llbracket \mathbf{id}_{\mathtt{V}_{m}}\rrbracket\} $ and $\{\llbracket \mathbf{d}_{\mathtt{V}_{m}}\rrbracket\} $.
		
		\STATE {\csa} initialize empty sets $\{\llbracket \mathbf{id}_{\mathtt{V}_{m}}\rrbracket\}, \{\llbracket \mathbf{d}_{\mathtt{V}_{m}}\rrbracket\} $.
		\IF{Case I}
		
		\STATE$\llbracket \mathbf{id}_{\mathtt{V}_{m}}\rrbracket=\bigoplus_{c=1}^{C}\llbracket \mathbf{id}_{\mathtt{V}_{c}}\rrbracket\otimes\llbracket x_{\mathtt{V}_{c}} \rrbracket; \{\llbracket \mathbf{id}_{\mathtt{V}_{m}}\rrbracket\}.\mathsf{add}(\llbracket \mathbf{id}_{\mathtt{V}_{m}}\rrbracket)$.\label{algo:case_I_1}
		
		\STATE$\llbracket \mathbf{d}_{\mathtt{V}_{m}}\rrbracket=\bigoplus_{c=1}^{C}\llbracket \mathbf{d}_{\mathtt{V}_{c}}\rrbracket\otimes\llbracket x_{\mathtt{V}_{c}} \rrbracket; \{\llbracket \mathbf{d}_{\mathtt{V}_{m}}\rrbracket\}.\mathsf{add}(\llbracket \mathbf{d}_{\mathtt{V}_{m}}\rrbracket)$. \label{algo:case_I_2}
		
		\ENDIF
		
		\IF{Case II}
		\STATE {\csa} regard the inputs as a secret-shared table $\llbracket \mathbf{D}\rrbracket=\{ \llbracket x_{\mathtt{V}_{c}}\rrbracket||\llbracket \mathbf{id}_{\mathtt{V}_{c}}\rrbracket||\llbracket \mathbf{d}_{\mathtt{V}_{c}}\rrbracket\} $.
		
		\STATE $\llbracket \widehat{\mathbf{D}}\rrbracket=\mathsf{secShuffle}(\llbracket \mathbf{D}\rrbracket)$.
		\STATE {\csa} open $\llbracket \hat{x}_{\mathtt{V}_{c}}\rrbracket\in\llbracket \widehat{\mathbf{D}}\rrbracket,c\in[1,C]$
		\FOR{$c\in[1,C]$}
		\IF{$\hat{x}_{\mathtt{V}_{c}}=1$}
		\STATE $\{\llbracket \mathbf{id}_{\mathtt{V}_{m}}\rrbracket\} .\mathsf{add}(\llbracket \widehat{\mathbf{id}}_{\mathtt{V}_{c}}\rrbracket)$; $\{\llbracket \mathbf{d}_{\mathtt{V}_{m}}\rrbracket\}.\mathsf{add}(\llbracket \widehat{\mathbf{d}}_{\mathtt{V}_{c}}\rrbracket)$.
		\ENDIF
		\ENDFOR
		\ENDIF
		
		\STATE {\csa}  output  $\{\llbracket \mathbf{id}_{\mathtt{V}_{m}}\rrbracket\} $ and $\{\llbracket \mathbf{d}_{\mathtt{V}_{m}}\rrbracket\} $.
	\end{algorithmic}
	
\end{algorithm}

The second case with two or more candidate vertices having $\llbracket x_{\mathtt{V}_{c}} \rrbracket=\llbracket 1 \rrbracket$ is complicated and demands delicate treatment.
%
% Case II, however, is hard to deal with because two or more candidate vertices have $\llbracket x_{\mathtt{V}_{c}} \rrbracket=\llbracket 1 \rrbracket$. 
%
Our key idea is to first have the cloud servers obliviously shuffle the candidate vertices' encrypted information $\{ \llbracket x_{\mathtt{V}_{c}}\rrbracket||\llbracket \mathbf{id}_{\mathtt{V}_{c}}\rrbracket||\llbracket \mathbf{d}_{\mathtt{V}_{c}}\rrbracket\} $ (``$||$'' denotes concatenation), i.e., a shuffle is performed without the cloud servers knowing the permutation.
Since the candidate vertices are shuffled, we can safely open the predicate evaluation results and identify which shuffled vertices are the matched ones.
Here what we need is a technique that can perform secure shuffling in the secret sharing domain.
In particular, given a secret-shared dataset with an ordered set of records $\llbracket\mathbf{D}\rrbracket=\{ \llbracket \mathbf{r}_{i}\rrbracket\}$ (named as \textit{table}; each record $\llbracket \mathbf{r}_{i}\rrbracket$ is a row in $\llbracket\mathbf{D}\rrbracket$ and can denote each candidate vertex's encrypted information in our context), we need a secret-shared shuffle protocol that allows the parties holding the shares to jointly \textit{shuffle} the records in $\llbracket\mathbf{D}\rrbracket$ and produce secret shares of the result $\llbracket\pi(\mathbf{D})\rrbracket$, while no party can learn the permutation $\pi(\cdot)$. 
We identify that the state-of-the-art protocol from \cite{araki2021secure} is well suited for our purpose, as it allows secret-shared shuffling in the RSS domain.
\revise{Algorithm \ref{alg:shuffle} shows the secret-shared shuffle protocol}, and we write $\llbracket\hat{\mathbf{D}}\rrbracket=\mathsf{secShuffle}(\llbracket\mathbf{D}\rrbracket)$ to denote the protocol.
{\main} adapts $\mathsf{secShuffle}$ to instantiate the above idea for handling the case of multiple matched vertices.

\begin{algorithm}[!t]
	\caption{\revise{Building Block: Secret Shuffling $\mathsf{secShuffle}$ \cite{araki2021secure}}}
\label{alg:shuffle}
\begin{algorithmic}[1] 
	\REQUIRE \revise{The ordered set of records $\llbracket \mathbf{D}\rrbracket^{B}$ in binary RSS; the seed of random value generator: {\cs} and {\css} hold $s_{12}$; {\css} and {\csss} hold $s_{23}$; {\csss} and {\cs} hold $s_{31}$.}
	
	\ENSURE \revise{The shuffled records $\llbracket\hat{\mathbf{D}}\rrbracket=\llbracket\pi(\mathbf{D})\rrbracket^{B}$.}

	\revise{\# {\csa} generate the pseudo-random permutations $\pi$ and  tables $\mathbf{T}, \mathbf{R}$ with the same size as $\llbracket \mathbf{D}\rrbracket^{B}$:
	\STATE {\cs} and {\css} use $s_{12}$ locally generate $\pi_{12}$, $\mathbf{T}_{12}$ and $\mathbf{R}_{2}$.
	\STATE {\css} and {\csss} use $s_{23}$ locally generate $\pi_{23}$ and $\mathbf{T}_{23}$.
	\STATE {\csss} and {\cs} use $s_{31}$ locally generate $\pi_{31}$, $\mathbf{T}_{31}$ and $\mathbf{R}_{1}$.
	\STATE {\cs}: $\mathbf{X}_{1}=\pi_{31}[\pi_{12}(\langle\mathbf{D}\rangle_{1}\oplus \langle\mathbf{D}\rangle_{2}\oplus \mathbf{T}_{12})\oplus \mathbf{T}_{31}]$; sends $\mathbf{X}_{1}$ to {\css}.
	\STATE{\css}: $\mathbf{Y}_{1}=\pi_{12}(\langle\mathbf{D}\rangle_{3}\oplus\mathbf{T}_{12})$; $\mathbf{C}_{1}=\pi_{23}(\mathbf{X}_{2}\oplus\mathbf{T}_{23})\oplus \mathbf{R}_{2}$; sends $\mathbf{Y}_{1}$ and $\mathbf{C}_{1}$ to {\csss}.
	\STATE {\csss}: $\mathbf{C}_{2}=\pi_{23}[\pi_{31}(\mathbf{Y}_{1}\oplus\mathbf{T}_{31})\oplus\mathbf{T}_{23}]\oplus\mathbf{R}_{1}$; $\mathbf{R}_{3}=\mathbf{C}_{1}\oplus\mathbf{C}_{2}$; sends $\mathbf{R}_{3}$ to {\css}.
	\STATE {\cs} holds $(\mathbf{R}_{1}, \mathbf{R}_{2})$, {\css} holds $(\mathbf{R}_{2},\mathbf{R}_{3})$, and {\csss} holds $(\mathbf{R}_{3},\mathbf{R}_{1})$ as the final secret shares of $\llbracket\pi(\mathbf{D})\rrbracket^{B}$.
}
\end{algorithmic}
\end{algorithm}

\begin{algorithm}[!t]
	\caption{Secure Neighboring Vertices Accessing $\mathsf{secAcc.}$}%(\{\llbracket \mathbf{id}_{\mathtt{V}_{m}}\rrbracket\} ,  T_{ne},t_{ne})$} 
\label{alg:neigh_access}
\begin{algorithmic}[1] 
	\REQUIRE The matched vertices' IDs $\{\llbracket \mathbf{id}_{\mathtt{V}_{m}}\rrbracket\} $; the neighboring vertices' type $T_{ne}$ and attribute type $t_{ne}$.
	
	\ENSURE Neighboring vertices' $\{\llbracket \mathbf{id}_{\mathtt{V}_{ne}}\rrbracket\} $, $\{\llbracket \mathbf{d}_{\mathtt{V}_{ne}}\rrbracket\} $.
	
	%	\STATE Find all candidate vertices' posting lists $\{\llbracket P^{T_{ne}}_{\mathtt{V}_{c}}\rrbracket\}$ and neighboring attribute values $\{\llbracket d_{t}\rrbracket\}_{t\in[1,T]}$.
	
	\FOR{$\mathtt{V}_{m}\in\{\mathtt{V}_{m}\}$}
	\STATE $		\llbracket P^{T_{ne}}_{\mathtt{V}_{m}}\rrbracket=\{\bigoplus_{c=1}^{C}\llbracket \mathbf{id}_{\mathtt{V}_{m}}[c]\rrbracket\otimes \llbracket \mathbf{id}_{c,l}\rrbracket\}_{l\in[ L_{max}]}$.

	%\STATE {\csa} regard $\llbracket P^{T_{ne}}_{\mathtt{V}_{m}}\rrbracket$ as a secret-shared table $\llbracket \mathbf{id}\rrbracket$ where each row is an $\llbracket id_{ne}\rrbracket$.
	\STATE {\csa} regard $\llbracket P^{T_{ne}}_{\mathtt{V}_{m}}\rrbracket$ as a table $\llbracket \mathbf{ID}\rrbracket$.
	
	\STATE  $\llbracket \widehat{\mathbf{ID}}\rrbracket=\mathsf{secShuffle}(\llbracket \mathbf{I D}\rrbracket)$.
	
	\FOR{$\llbracket \widehat{\mathbf{id}}_{\mathtt{V}_{ne}}\rrbracket\in\llbracket \widehat{\mathbf{ID}}\rrbracket$}
	
	\STATE Open $\llbracket y_{\mathtt{V}_{ne}}\rrbracket=\bigoplus_{x=1}^{X}\llbracket \widehat{\mathbf{id}}_{\mathtt{V}_{ne}}[x]\rrbracket$.
	
	\IF{$y_{\mathtt{V}_{ne}}==1$}
	
	\STATE Add $\llbracket \widehat{\mathbf{id}}_{\mathtt{V}_{ne}}\rrbracket$ into the outputs $\{\llbracket \mathbf{id}_{\mathtt{V}_{ne}}\rrbracket\}$.
	
	\STATE 		$\llbracket \mathbf{d}_{\mathtt{V}_{ne}}\rrbracket=\bigoplus_{x=1}^{X}\llbracket \mathbf{id}_{\mathtt{V}_{ne}}[x]\rrbracket\otimes \llbracket \mathbf{d}_{x}\rrbracket$.
	
	\STATE Add $\llbracket \mathbf{d}_{\mathtt{V}_{ne}}\rrbracket$ into the outputs $\{\llbracket \mathbf{d}_{\mathtt{V}_{ne}}\rrbracket\}$.
	\ENDIF
	\ENDFOR
	\ENDFOR
	
	\STATE {\csa} output $\{\llbracket \mathbf{id}_{\mathtt{V}_{ne}}\rrbracket\} $ and $\{\llbracket \mathbf{d}_{\mathtt{V}_{ne}}\rrbracket\} $.
\end{algorithmic}
\end{algorithm}

\noindent\textbf{Secure neighboring vertices accessing.} 
With the encrypted ID $\llbracket \mathbf{id}_{\mathtt{V}_{m}} \rrbracket$ produced for each matched vertex $\mathtt{V}_{m}$, {\csa} then need to access information of each matched vertex's neighboring vertices, which are used as the candidate vertices for the next-hop target vertex in $\mathsf{tok}_{q}$. 
We denote the neighboring vertices of each $\mathtt{V}_{m}$ by $\{\mathtt{V}_{ne}\}$.
%
% Simply opening the ID for each matched vertex to identify its neighboring vertices will easily violate the security requirement for access pattern protection.
%
{\main} devises a component $\mathsf{secAccess}$, as shown in Algorithm \ref{alg:neigh_access}, to allow {\csa} to \emph{obliviously} access the information about the neighboring vertices.
Specifically, {\csa} should be able to access $\{\mathtt{V}_{ne}\}$'s IDs $\{\llbracket \mathbf{id}_{\mathtt{V}_{ne}} \rrbracket\}$ and attribute values $\{\llbracket \mathbf{d}_{\mathtt{V}_{ne}} \rrbracket\}$, without knowing which vertices in $\llbracket\mathcal{G}^{k}\rrbracket$ they are.

For ease of presentation, we only consider that each target vertex in $\mathsf{tok}_{q}$ only has one target neighboring vertex.
The support for multiple target neighboring vertices is straightforward, where  {\csa} handle with each of them independently. 
We first introduce how {\csa} obliviously  fetch each $\mathtt{V}_{m}$'s neighboring vertices' IDs $\{\llbracket \mathbf{id}_{\mathtt{V}_{ne}} \rrbracket\}$ via $\mathtt{V}_{m}$'s ID $\llbracket \mathbf{id}_{\mathtt{V}_{m}}\rrbracket$.
It is noted that the type of neighboring vertices $\{\mathtt{V}_{ne}\}$ is $T_{ne}$, i.e., the type of the next-hop target vertex in $\mathsf{tok}_{q}$. 
Therefore, $\mathtt{V}_{m}$'s posting list  $\llbracket P^{T_{ne}}_{\mathtt{V}_{m}}\rrbracket$ with type $T_{ne}$ contains the needed $\{\llbracket \mathbf{id}_{\mathtt{V}_{ne}}\rrbracket\}$. %
{\csa} should obliviously fetch $\llbracket P^{T_{ne}}_{\mathtt{V}_{m}}\rrbracket$ from all candidate vertices $\{\mathtt{V}_{c}\}$'s posting lists $\{\llbracket P^{T_{ne}}_{\mathtt{V}_{c}}\rrbracket\}$.
Our key insight is to utilize the benefits that $\mathtt{V}_{m}$'s ID $\llbracket \mathbf{id}_{\mathtt{V}_{m}}\rrbracket$ is encoded as a one-hot vector and protected via RSS. 
Specifically, {\main} lets {\csa} obliviously AND each bit $\llbracket \mathbf{id}_{\mathtt{V}_{m}}[c]\rrbracket,c\in[C]$ by each candidate vertex $\mathtt{V}_{c}$'s posting list $\llbracket P^{T_{ne}}_{\mathtt{V}_{c}}\rrbracket$, and then XOR the AND operation results to obtain $\mathtt{V}_{m}$'s posting list $\llbracket P^{T_{ne}}_{\mathtt{V}_{m}}\rrbracket $. 
Formally, {\csa} perform the following:
\begin{equation}\notag
\begin{array}{rl}
	\llbracket P^{T_{ne}}_{\mathtt{V}_{m}}\rrbracket=\{\bigoplus_{c=1}^{C}\llbracket \mathbf{id}_{\mathtt{V}_{m}}[c]\rrbracket\otimes \llbracket \mathbf{id}_{c,l}\rrbracket\}_{l\in[ L_{max}]},
\end{array}
\end{equation}
where $\llbracket\mathbf{id}_{c,l}\rrbracket$ is the $l$-th ID in $\mathtt{V}_{c}$'s posting list $\llbracket P^{T_{ne}}_{\mathtt{V}_{c}}\rrbracket$ and $L_{max}$ is the maximum length of all candidate vertices' posting lists. Correctness holds since there is only one  $ 1$  in the one-hot vector $ \mathbf{id}_{\mathtt{V}_{m}}$, whose location corresponds to $\mathtt{V}_{m}$'s location in $\mathcal{G}$, and thus only the IDs in $\mathtt{V}_{m}$'s $ P^{T_{ne}}_{\mathtt{V}_{m}}$ will be kept.

However, since the lengths of different candidate vertices'  posting lists are varying and there are also some dummy IDs in some posting lists (to achieve $k$-automorphism when encrypting the attributed graph as introduced in Section \ref{sec:enc}), the fetched $\llbracket P^{T_{ne}}_{\mathtt{V}_{m}}\rrbracket$ may contain some invalid IDs, which will incur undesirable performance overheads. 
Therefore, {\main} lets {\csa} further obliviously refine $\llbracket P^{T_{ne}}_{\mathtt{V}_{m}}\rrbracket$ to filter out these invalid IDs.
We observe that the invalid IDs are 0-vectors, and thus {\main} lets {\csa} first locally XOR each bit of $\llbracket \mathbf{id}_{\mathtt{V}_{ne}}\rrbracket\in \llbracket P^{T_{ne}}_{\mathtt{V}_{m}}\rrbracket$:
\begin{equation}\notag
\begin{array}{rl}
	\llbracket y_{\mathtt{V}_{ne}}\rrbracket=\bigoplus_{x=1}^{X}\llbracket \mathbf{id}_{\mathtt{V}_{ne}}[x]\rrbracket,
\end{array}
\end{equation}
where $X$ is the length of $\mathbf{id}_{\mathtt{V}_{ne}}$, which is the number of vertices with type $T_{ne}$ in $\llbracket\mathcal{G}^{k}\rrbracket$.
$ y_{\mathtt{V}_{ne}}=0$ indicates that $ \mathbf{id}_{\mathtt{V}_{ne}}$ is a 0-vector and an invalid ID.
After that,  a naive method is to let {\csa} open each $\llbracket y_{\mathtt{V}_{ne}}\rrbracket$ to judge whether its corresponding ID is invalid. 
However, the naive method may leak the search pattern (recall Definition \ref{def:query_pattern}) since the orders of $\{\llbracket \mathbf{id}_{\mathtt{V}_{ne}}\rrbracket\}$ in $ \llbracket P^{T_{ne}}_{\mathtt{V}_{m}}\rrbracket$ is static, which makes the same query undoubtedly yield the same opening results $\{y_{\mathtt{V}_{ne}}\}$.

Our solution is to let {\csa} obliviously permute the encrypted IDs in $\llbracket P^{T_{ne}}_{\mathtt{V}_{m}}\rrbracket$ before opening $\{\llbracket y_{\mathtt{V}_{ne}}\rrbracket\}$.
Since the encrypted IDs are shuffled, we can safely open $\{\llbracket y_{\mathtt{V}_{ne}}\rrbracket\}$ and identify which shuffled $\llbracket \mathbf{id}_{\mathtt{V}_{ne}}\rrbracket$ are the invalid IDs.
Specifically, {\main} lets {\csa} regard $\llbracket P^{T_{ne}}_{\mathtt{V}_{m}}\rrbracket$ as a table, where each $\llbracket \mathbf{id}_{\mathtt{V}_{ne}}\rrbracket\in\llbracket P^{T_{ne}}_{\mathtt{V}_{m}}\rrbracket$ is a record, and then obliviously shuffle $\llbracket P^{T_{ne}}_{\mathtt{V}_{m}}\rrbracket$ via $\mathsf{secShuffle}$ followed by opening the XOR operation results $\{\llbracket y_{\mathtt{V}_{ne}}\rrbracket\}$ to filter out the invalid IDs.
By this way, {\csa} can obliviously obtain the accurate $\llbracket P^{T_{ne}}_{\mathtt{V}_{m}}\rrbracket$, without knowing the search pattern.

After that, {\csa} should obliviously fetch the encrypted value $\llbracket \mathbf{d}_{\mathtt{V}_{ne}} \rrbracket$ of each neighboring vertex $\mathtt{V}_{ne}$'s attribute with type $t_{ne}$ via its ID $\llbracket \mathbf{id}_{\mathtt{V}_{ne}} \rrbracket$, where $t_{ne}$ is the type of the target attribute associated with the next-hop target vertex in $\mathsf{tok}_{q}$. 
Our key insight is to utilize the benefits that $\llbracket \mathbf{id}_{\mathtt{V}_{ne}}\rrbracket\in\llbracket P^{T_{ne}}_{\mathtt{V}_{m}}\rrbracket$ is encoded as a one-hot vector and protected via RSS as before. 
Specifically, {\main} lets {\csa} first locally retrieve the vertices with type $T_{ne}$ from $\llbracket\mathcal{G}^{k}\rrbracket$, and then locally retrieve the values of their attributes with type $t_{ne}$, denoted as $\{\llbracket\mathbf{d}_{x}\rrbracket\}_{x\in[X]}$.
After that, {\csa} obliviously AND each bit $\llbracket \mathbf{id}_{\mathtt{V}_{ne}}[x]\rrbracket,x\in[X]$ by $\llbracket \mathbf{d}_{x}\rrbracket$, and then XOR the AND operation results to obtain $\mathtt{V}_{ne}$'s attribute value $\llbracket \mathbf{d}_{\mathtt{V}_{ne}}\rrbracket$.
Formally, {\csa} perform the following:
\begin{equation}\notag
\begin{array}{rl}
	\llbracket \mathbf{d}_{\mathtt{V}_{ne}}\rrbracket=\bigoplus_{x=1}^{X}\llbracket \mathbf{id}_{\mathtt{V}_{ne}}[x]\rrbracket\otimes \llbracket \mathbf{d}_{x}\rrbracket.
\end{array}
\end{equation}

Then all matched vertices' neighboring vertices' IDs and attribute values compose $\{\llbracket \mathbf{id}_{\mathtt{V}_{ne}}\rrbracket\}$ and $\{\llbracket \mathbf{d}_{\mathtt{V}_{ne}}\rrbracket\}$.
Finally, {\csa} set $\{\llbracket \mathbf{id}_{\mathtt{V}_{ne}}\rrbracket\}$ and $\{\llbracket \mathbf{d}_{\mathtt{V}_{ne}}\rrbracket\}$ as new candidate vertices' $\{\llbracket \mathbf{id}_{\mathtt{V}_{c}}\rrbracket\}$ and $\{\llbracket \mathbf{d}_{\mathtt{V}_{c}}\rrbracket\}$ for the use in the next-hop target vertex matching.

\section{Security Analysis}
\label{sec:security_analysis}

\iffalse

In this section, we provide analysis regarding the protection {\main} offers for the attributed subgraph matching system. 
%
Follow the threat model defined in Section \ref{sec:threat_model}, we consider a PPT adversary $\mathcal{A}$, who corrupts at most one of {\csa}. 
%
Under $\mathcal{A}$'s control, the corrupted cloud server honestly follows our protocol, but may \textit{individually} attempts to infer the private information. 
%
\fi

We follow the simulation-based paradigm \cite{lindell2017simulate} to prove the security guarantees of {\main}. 
We start with defining the ideal functionality $\mathcal{F}$ for oblivious and encrypted attributed subgraph matching, which comprises the following parts:
\begin{itemize}
\item \textbf{Input.} {\cli} submits the attributed graph $\mathcal{G}$ and a subgraph query $q$ to $\mathcal{F}$.

\item \textbf{Computation.} Upon receiving $\mathcal{G}$ and $q$ from {\cli}, $\mathcal{F}$ retrieves the subgraphs $ \{ g_{m}\}$ isomorphic to $q$ from $\mathcal{G}$. %After that,  $\mathcal{F}$ ranks $ \{\llbracket g_{m}\rrbracket\}$ to $ \{\llbracket \hat{g}_{m}\rrbracket\}$.

\item \textbf{Output.} $\mathcal{F}$ returns subgraphs $ \{ g_{m}\}$ to {\cli}.
\end{itemize}

We allow $\mathcal{F}$ to leak $\mathsf{leak}(\mathcal{F})=(\mathtt{schema}^{\{\mathcal{G},q\}}, \mathtt{struct}^{q})$ as defined in Section \ref{sec:threat_model}, where $\mathtt{schema}^{\{\mathcal{G},q\}}$ are the schema layout parameters  of $\mathcal{G}$ and $q$ and $\mathtt{struct}^{q}$ is the structure of $q$.
Let $\prod$ denote a protocol for secure attributed subgraph matching realizing the ideal functionality $\mathcal{F}$,  $\prod$'s security is formally defined as follows.
\begin{definition}
\label{def:security}
Let $\mathcal{A}$ be an adversary who observes the view of a corrupted server during $\prod$'s execution. Let $\mathsf{View}_{\prod(\mathcal{A})}^{\mathsf{Real}}$ denote $\mathcal{A}$'s view in the real world experiment. 
In the ideal world, a simulator $\mathcal{S}$ generates a simulated view $\mathsf{View}_{\mathcal{S}, \mathsf{leak}(\mathcal{F})}^{\mathsf{Ideal}}$ to $\mathcal{A}$ given only the leakage $\mathsf{leak}(\mathcal{F})$. 
After that, $\forall$ PPT adversary $\mathcal{A}$, $\exists$ a PPT simulator $\mathcal{S}$ s.t. $\mathsf{View}_{\prod(\mathcal{A})}^{\mathsf{Real}}\mathop \approx\mathsf{View}_{\mathcal{S}, \mathsf{leak}(\mathcal{F})}^{\mathsf{Ideal}}$.
\end{definition}

\begin{theorem}
\label{theo:security}
According to Definition \ref{def:security}, {\main} can securely realize the ideal functionality $\mathcal{F}$ when instantiated with secure DPFs, DCFs, secret shuffling and a pseudo-random function, assuming a semi-honest and non-colluding adversary model. 
\end{theorem}
\begin{proof}
{\main} consists of three secure subroutines: 1) attributed graph encryption $\mathsf{encGraph}$; 2) secure query token generation $\mathsf{genToken}$; 3) secure attributed graph matching $\mathsf{secMatch}$. 
Each subroutine in {\main} is invoked in order as per the processing pipeline and their inputs are secret shares.  
Therefore, if the simulator for each subroutine exists, then our complete protocol is secure \cite{canetti2000security,katz2005handling,curran2019procsa}. 
Since the roles of {\csa} in these subroutines are symmetric, it suffices to show the existence of simulators for {\cs}.
\begin{itemize}
	\item \textbf{Simulator for {\cs} in $\mathsf{encGraph}$.} Since {\cs} only receives the RSS-based secret shares during $\mathsf{encGraph}$, the simulator for $\mathsf{encGraph}$ can be trivially constructed by invoking the RSS simulator. 
	Therefore, from the security of RSS \cite{araki2016high}, the simulator for $\mathsf{encGraph}$ exists.

	\item \textbf{Simulator for {\cs} in $\mathsf{genToken}$.} Since {\cs} only  receives FSS keys (i.e., $\{\mathcal{K}_{i}\}$) apart from the public information, %schema layout parameters $\mathtt{schema}^{q}$ and the structure $\mathtt{struct}^{q}$,  
	the simulator for $\mathsf{genToken}$ can be trivially constructed by invoking the FSS simulator. 
	Therefore, from the security of FSS \cite{boyle2016function,boyle2021function}, the simulator for $\mathsf{genToken}$ exists.
	
	\item \textbf{Simulator for {\cs} in $\mathsf{secMatch}$.} It is noted that $\mathsf{secMatch}$ (i.e., Algorithm \ref{alg:main}) consists of three components and each component is invoked in order as per the processing pipeline. 
	We analyze the existence of their simulators in turn:% as follows.
	\begin{itemize}
		\item  \textbf{Simulator for {\cs} in $\mathsf{secEval}$.} 
		Since in each function loop of $\mathsf{secEval}$ (i.e., Algorithm \ref{alg:select}), {\cs} evaluates FSS keys on the independent secret shares, we only analyze the existence of simulator for one function loop. 
		At the beginning of a function loop, {\cs} has two FSS keys $\langle k_{1}^{1}\rangle, \langle k_{1}^{2}\rangle$ and secret shares $\langle \mathbf{d}_{\mathtt{V}_{c}}\rangle_{1}, \langle \mathbf{d}_{\mathtt{V}_{c}}\rangle_{2}$, later outputs the evaluation results $x^{(1)}_{1}, x^{(2)}_{1}$. 
		Since these information {\cs} views is all legitimate in FSS,  the simulator for the evaluation can be trivially constructed by invoking the simulator of FSS. 
		After that, {\cs} receives a secret share from {\css}, i.e., re-sharing in RSS, and thus the simulator for the re-sharing can be trivially constructed by invoking the RSS simulator. 
		Therefore, from the security of FSS \cite{boyle2016function,boyle2021function} and RSS \cite{araki2016high}, the simulator for $\mathsf{secEval}$ exists.
		
		\item  \textbf{Simulator for {\cs} in $\mathsf{secFetch}$.} It is noted that there are two cases in $\mathsf{secFetch}$ (i.e., Algorithm \ref{alg:reveal}).
		Since case I consists of basic operations (i.e., $\oplus$ and $\otimes$) in RSS,  the simulator for it can be trivially constructed by invoking the simulator of RSS. 
		At the beginning of case II, {\cs} has secret shares $ \{ \langle x_{\mathtt{V}_{c}}\rangle_{1}||\langle \mathbf{id}_{\mathtt{V}_{c}}\rangle_{1}||\langle \mathbf{d}_{\mathtt{V}_{c}}\rangle_{1}\}$ , $\{ \langle x_{\mathtt{V}_{c}}\rangle_{2}||\langle \mathbf{id}_{\mathtt{V}_{c}}\rangle_{2}||\langle \mathbf{d}_{\mathtt{V}_{c}}\rangle_{2}\} $, and later receives secret shares in secure shuffling and secret shares $\{\langle \hat{x}_{\mathtt{V}_{c}}\rangle_{3}\} $ from {\csss} to recover $\{\hat{x}_{\mathtt{V}_{c}}\} $. 
		%
		%The simulator for the secret shuffling can be trivially constructed by invoking the secret shuffling simulator. % and the simulator for the re-sharing can be trivially constructed by invoking the RSS simulator. 
		%
		Therefore, from the security of secret shuffling \cite{araki2021secure} and RSS \cite{araki2016high}, the simulator for $\mathsf{secFetch}$ exists. 
		
		\item  \textbf{Simulator for {\cs} in $\mathsf{secAccess}$.} Similar to the analysis for $\mathsf{secFetch}$, the simulator for $\mathsf{secAccess}$ (i.e., Algorithm \ref{alg:neigh_access}) can be trivially constructed by invoking the simulator of RSS and the simulator of secure shuffle $\mathsf{secShuffle}$. 
		Therefore, from the security of secret shuffling \cite{araki2021secure} and RSS \cite{araki2016high}, the simulator for $\mathsf{secAccess}$ exists.
	\end{itemize}
	%\item  \textbf{Simulator for {\cs} in $\mathsf{secRank}$.} At the beginning of $\mathsf{secRank}$, {\cs} has the shares of key attribute values $\{\langle d_{m}\rangle_{1}\},\{\langle d_{m}\rangle_{2}\}$, and then receives secret shares during the transformation from the binary RSS to arithmetic RSS from {\csss}. 
	%
	%After that, {\cs} evaluates DCF keys  on the ciphertext followed by opening the comparison results. 
	%
	%The simulator for $\mathsf{secRank}$ can be trivially constructed by invoking the RSS transformation simulator and the RSS simulator. 
	%
	%Therefore, from the security of RSS transformation \cite{mohassel2018aby3} and RSS \cite{araki2016high}, the simulator for $\mathsf{secRank}$ exists. 
\end{itemize}

We now explicitly analyze why {\main} can hide search access patterns as follows.
\begin{itemize}
	\item  \textbf{Hiding the search pattern.} Given a query token, each {\csa} only receives the FSS keys (i.e., $\{\mathcal{K}_{i}\}$) apart from the public schema layout parameters and the structure of the query. 
	The security of FSS guarantees that even encrypting the same value multiple times will result in different FSS keys indistinguishable from uniformly random values. 
	Therefore, from the security of FSS \cite{boyle2016function,boyle2021function}, {\csa} cannot determine whether a new query has been issued before (except knowing whether the public structure was used before). 
	In addition, in the process of secure subgraph matching,  {\main} lets {\csa} shuffle the (binary) evaluation results $\{\llbracket x_{\mathtt{V}_{c}}\rrbracket\}$ and $\{\llbracket y_{\mathtt{V}_{ne}}\rrbracket\}$ before opening them. 
	From the security of secure shuffle \cite{araki2021secure}, even processing the same queries multiple times will result in different orders of the opened results. Since $\{x_{\mathtt{V}_{c}}\}$ and $\{y_{\mathtt{V}_{ne}}\}$ are bit-strings, secure shuffle ensures that even processing the same query multiple times will result in different opened bit-strings at each time.
	So these opened binary evaluation results will not indicate whether two queries are the same or not.
	Therefore, {\main} can hide the search pattern.

	\item  \textbf{Hiding the access pattern.} As per Definition \ref{def:access_pattern}, the access pattern in fact indicates whether a vertex in the encrypted attributed graph $\llbracket\mathcal{G}^{k}\rrbracket$ is a matched vertex, namely, whether it will appear in the matching results $ \{\llbracket g_{m}\rrbracket\}$. 
	Since the matched vertices are (obliviously) determined in $\mathsf{secFetch}$, we only need to analyze $\mathsf{secFetch}$. Recall there are two cases in $\mathsf{secFetch}$.
	For case I, it does not leak the access pattern apparently since all processing is in secret sharing domain and nothing is opened. For case II, before opening the evaluation results $\{\llbracket x_{\mathtt{V}_{c}}\rrbracket\} $ of candidate vertices, {\main} lets {\csa} obliviously shuffle $\{\llbracket x_{\mathtt{V}_{c}}\rrbracket||\llbracket \mathbf{id}_{\mathtt{V}_{c}}\rrbracket||\llbracket \mathbf{d}_{\mathtt{V}_{c}}\rrbracket\} $, which breaks the mapping relationship between the candidate vertices and $\llbracket\mathcal{G}^{k}\rrbracket$. 
	Therefore, from the security of secure shuffle \cite{araki2021secure},  {\main} can hide the access pattern.
	
\end{itemize} 
The proof of Theorem \ref{theo:security} is completed.
\end{proof}

\noindent\textbf{Discussion.} Attacks exploiting search access patterns have received wide attention and hiding these patterns is crucially important, which {\main} ambitiously explores and provides corresponding guarantees.
Additionally, we note that there are emerging volume-based attacks \cite{grubbs2018pump,gui2019encrypted,kornaropoulos2021response} exploiting the volume of results.
However, their dedicated assumptions make the attacks ineffective in our context. 
Specifically, the works \cite{grubbs2018pump, gui2019encrypted} assume that the database is dense, i.e., there is at least one record for every possible value of the plaintext domain, which obviously cannot be achieved in attributed graphs.
The work \cite{kornaropoulos2021response} assumes that the adversary issues independent and identically distributed queries with respect to a fixed query distribution and also does not address encrypted databases for high-dimensional data, which also does not stand in attributed graphs because of the heterogeneity.

\section{Performance Evaluation}
\label{sec:exp}

\subsection{Setup}

We implement a prototype system of {\main} in C++. 
Our prototype implementation comprises $\sim$1500 lines of code (excluding the code of libraries). 
We also implement a test module with another $\sim$300 lines of code. 
Three Alibaba Cloud ECS c8g1.2xlarge instances are used to act as {\csa}, each has a NVIDIA Tesla V100 GPU with 16 GB memory. 
All of the instances run Ubuntu 20.04 and have 8 Intel Platinum 8163 CPU cores and 32 GB of RAM. 
In addition, a Macbook Air with 8 GB RAM acts as {\cli} to generate and send query tokens. 
For the adopted cloud environment, the network bandwidth is 2.5 Gbps with an average latency of 0.2 ms.

%We consider a local area network (LAN) environment and place all three servers in the same data center. 
%Similar to prior work \cite{tan2021cryptgpu} in the multi-server setting, the network bandwidth is 1.25 Gbps with an average latency of 0.2 ms.

\noindent\textbf{Protocol instantiation.} Note that the all private data in {\main} is encrypted in binary RSS, and thus we can store these data in bit-strings $\{0,1\}^{n}$. 
However, $\mathsf{bool}$ data (i.e., $\{\mathsf{true},\mathsf{false}\}$) in C++ is stored as an 8-bit data, which will incur undesirable storage overheads when storing bit-strings. 
Therefore, we divide each private bit-string into 32-bit sub-vectors, and store each sub-vector in a 32-bit $\mathsf{unsigned  ~int}$ to save the storage. 
In addition, for DPFs and DCFs, we set the security parameter $\lambda$ to 128.

\noindent\revise{
\noindent\textbf{Performance boost from GPU.} We note that the overall design of {\main} is highly \textit{parallelizable}, which enables us to take advantage of GPU for parallel processing to achieve a performance boost.
Specifically, the GPU architecture is optimized for performing a large number of simple computations on blocks of values, which means that operations like component-wise addition and multiplication of vectors/matrices on GPU can be executed fast \cite{pandey2021trust}. 
It is noted that the main operations in {\main} are the addition and multiplication between (one-hot) vectors. 
In addition, the secure predicate evaluation is also parallelizable because each evaluation is independent, which enables us to allocate the secure predicate evaluation of different attribute values on independent GPU cores for parallel processing. 
Therefore, we implement {\main} utilizing the optimized NVIDIA's C++ based CUDA kernels. 
However, we note that the size of the attributed graph ciphertext may exceed the GPU memory.
Therefore, we only load necessary ciphertext (e.g., targeted attribute values by the queries for the secure predicate evaluation) in the GPU memory, instead of loading the complete encrypted  attributed graph.
We use the library $\mathsf{cuRAND}$\footnote{\url{https://docs.nvidia.com/cuda/curand/index.html}} to generate random values on GPU.
}

\noindent\textbf{Dataset.} We use an attributed graph dataset from \cite{leskovec2012learning}, which contains 107614 vertices and 13673453 edges. 
From the dataset, we exact two types of vertices ``$\mathtt{university}$'' and ``$\mathtt{person}$'',  two types of edges ``$\mathtt{graduate ~from}$'' and ``$\mathtt{friend}$'' and two types of attributes ``$\mathtt{location}$'' and ``$\mathtt{age}$''.

\vspace{-10pt}

\subsection{Evaluation on Attributed Graph Encryption}

Recall that {\cli} needs to model the attributed graph, add some dummy IDs into posting lists, and split the private information into binary RSS. 
The time and storage cost of encrypting the dataset under different $k\in\{2,4,6\}$ (i.e., $k$-automorphism) are $\{54,66,69\}$ minutes and $\{161,179,198\}$ GB, respectively. 
It is worth noting that such pre-processing cost is \textit{one-off} and does not affect the online service quality.
\vspace{-10pt}
\subsection{Evaluation on Secure Query Token Generation}

\begin{table}[!t]
\centering
\small
\caption{Time Cost (s) and Token Size (MB) under Different Subgraph Queries ($\ll$ Indicates Interval-based Range Query)} 		 	
\label{tab:token_gen}
\begin{tabular}{c|ccc|ccc}
	\hline
	&  \multicolumn{3}{c|}{Time cost (s)}&   \multicolumn{3}{c}{Token size (MB)} \\\hline
	&$|q|=$2&$|q|=$4&$|q|=$8&$|q|=$2&$|q|=$4&$|q|=$8\\\hline
	=&0.03&0.06&0.12&0.07&0.14&0.28\\
	$<$&0.04&0.08&0.16&0.07&0.14&0.28\\
	$\ll$&0.08&0.16&0.32&0.14&0.28&0.56\\\hline
\end{tabular}
\vspace{-10pt}
\end{table}

Recall that {\cli} needs to parse a subgraph query into the corresponding secure query token.
In particular, {\cli} should generate FSS keys for the value information of each private predicate in the query.
We conduct experiments with varying predicate types (i.e., $=,<$ and $\ll$) and the number of target vertices in subgraph queries (i.e., $|q|\in\{2,4,8\}$), and summarize the time cost and token size in Table \ref{tab:token_gen}.

\subsection{Performance Benchmarks on Sub-Protocols}
% In this section, to understand the performance of secure attributed subgraph matching in {\main}, 

We provide the performance benchmarks of {\main}'s three sub-protocols under different data sizes, i.e., $\mathsf{secEval}$, $\mathsf{secFetch}$ and $\mathsf{secAccess}$. 

\begin{table}
\centering
\small
\caption{Time Cost (s) of Sub-protocols under Different Data Sizes}\label{tab:subprotocol}
\begin{tabular}{c|ccc|cc|c}
	\hline
	& \multicolumn{3}{c|}{$\mathsf{secEval}$}& \multicolumn{2}{c|}{$\mathsf{secFetch}$}&	$\mathsf{secAccess}$\\\hline
	Size&  =&  $<$&  $\ll$&Case I&Case II&$\sim$\\ \cline{2-7}
	1000&$1.4$&1.6&1.8&0.1&1.2&2.1\\
	5000&$2.1$&2.3&2.5&0.1&1.3&2.5\\
	10000&$2.6$&2.9&3.4&0.3&1.6&2.9\\
	\hline
\end{tabular}
\vspace{-10pt}
\end{table}

%\begin{table}
%	\centering
%	\caption{Time cost (s) of sub-protocols under different data sizes.}\label{tab:subprotocol}
%	\begin{tabular}{c|ccc|cc|c|c}
%		\hline
%		& \multicolumn{3}{c|}{$\mathsf{secEval}$}& \multicolumn{2}{c|}{$\mathsf{secFetch}$}&	$\mathsf{secAc.}$&$\mathsf{secRank}$\\\hline
%		Size&  =&  $<$&  $\ll$& I& II&$\sim$&$\sim$\\\cline{2-8}
%		1000&$1.4$&1.6&1.8&0.1&1.2&2.1&3.1\\
%		5000&$2.1$&2.3&2.5&0.1&1.3&2.5&8.8\\
%		10000&$2.6$&2.9&3.4&0.3&1.6&2.9&17.1\\
%		\hline
%	\end{tabular}
%	\vspace{-10pt}
%\end{table}

\noindent\textbf{Computational efficiency.} We first evaluate the computational efficiency of each sub-protocol, with results provided in Table \ref{tab:subprotocol}. 
%
% conduct experiments under $k=6$ (i.e., $k$-automorphism) and data size $\in\{1000,5000,10000\}$ to 
%
From the results, it can be observed that the time cost of the three modules
%$\mathsf{secEval}$ ($=,<,\ll$), $\mathsf{secFetch}$ (case I and II) and $\mathsf{secAccess}$ 
are all not linear in the data size. 
This benefits from the high parallelizability of {\main}, which enables different GPU cores to perform independent sub-tasks simultaneously, e.g., securely evaluating the same private predicate on different encrypted attribute values or securely accessing different matched vertices' neighboring vertices.

\begin{figure}[t!]
\centering
\begin{minipage}[t]{0.32\linewidth}
	\centering{\includegraphics[width=\linewidth]{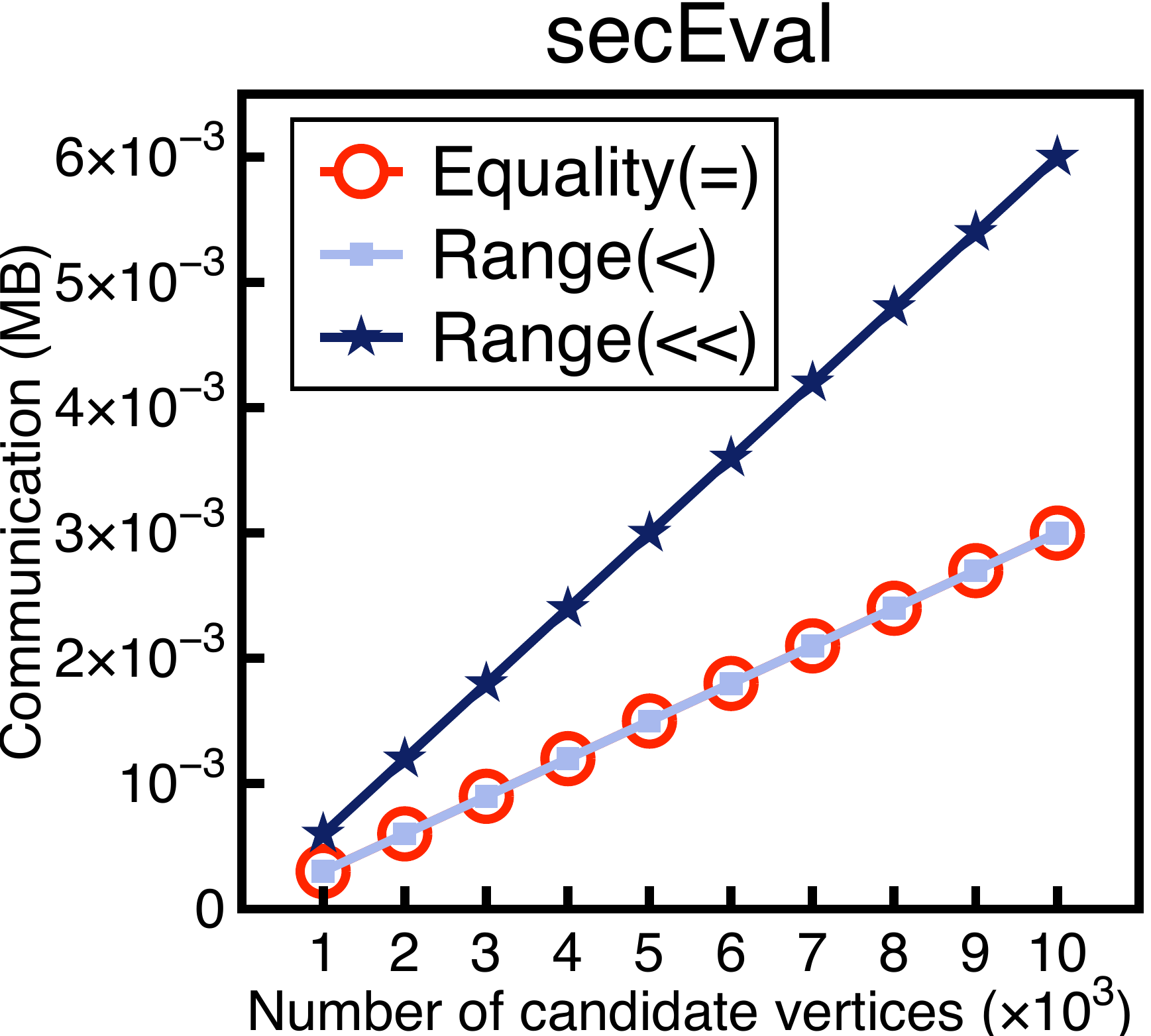}}
\end{minipage}
\begin{minipage}[t]{0.32\linewidth}
	\centering{\includegraphics[width=\linewidth]{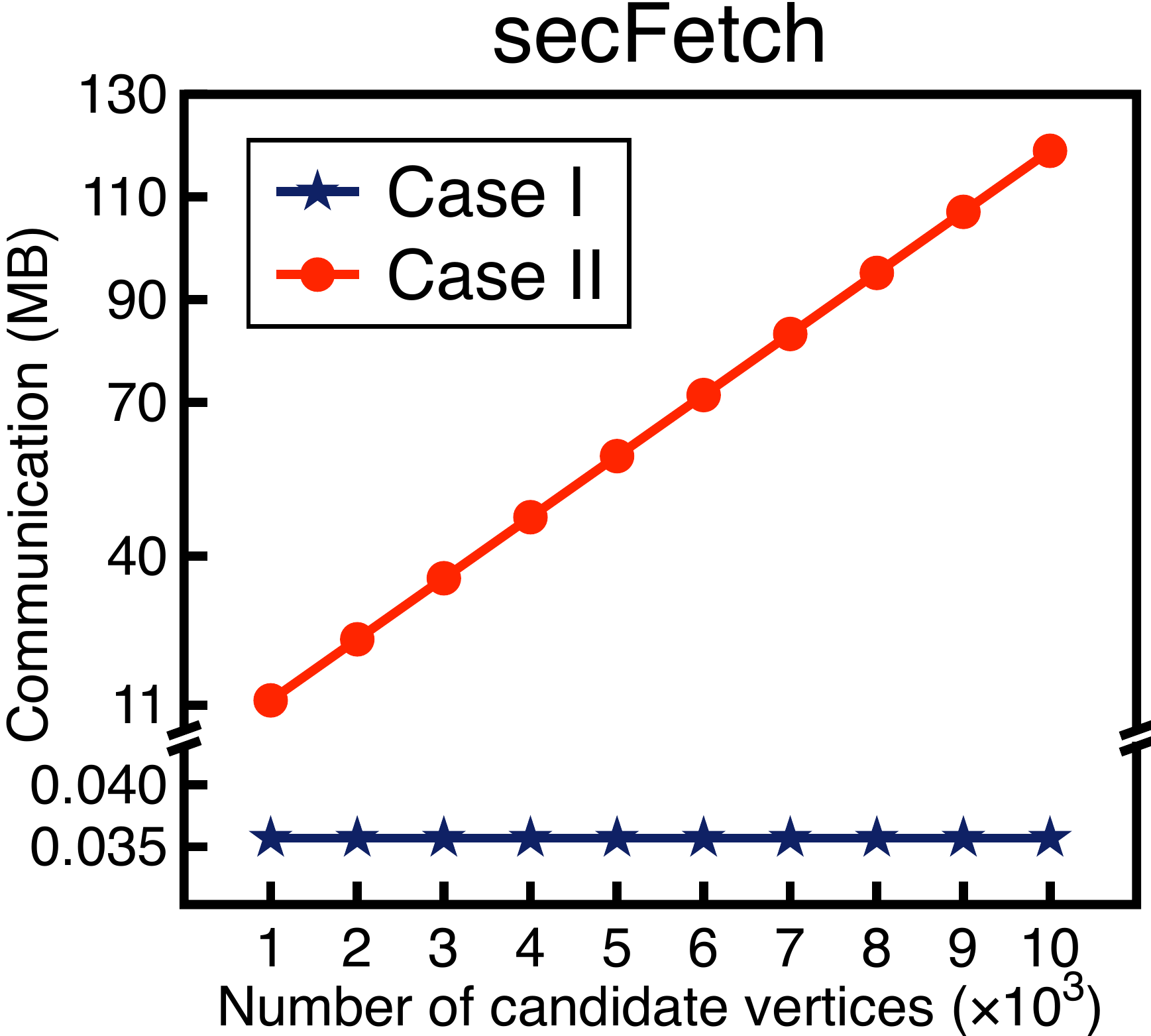}}
\end{minipage}	
\begin{minipage}[t]{0.32\linewidth}
	\centering{\includegraphics[width=\linewidth]{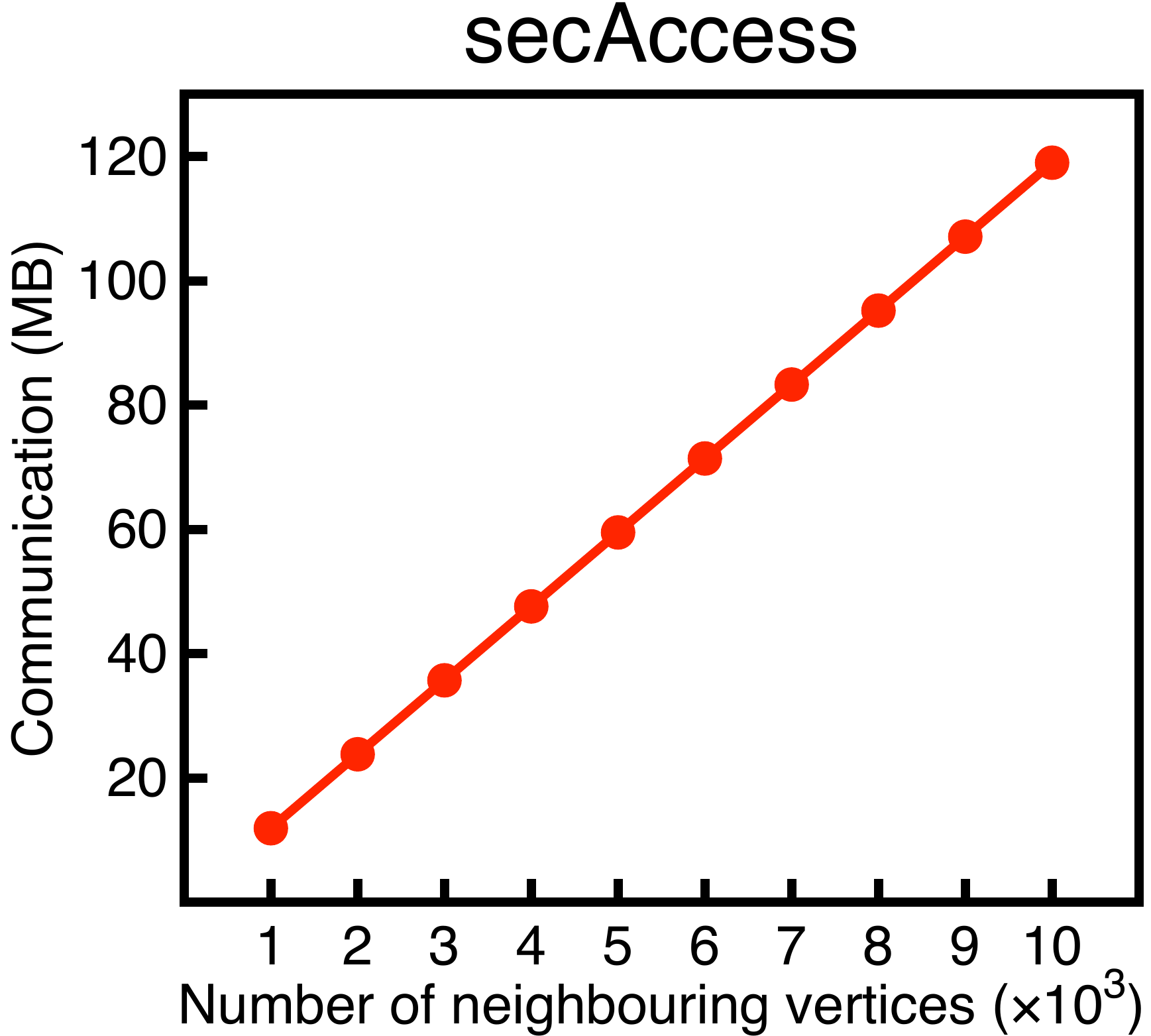}}
\end{minipage}

\caption{Communication of sub-protocols under different data sizes.}
\label{fig:subprotocol}
\vspace{-15pt}
\end{figure}

\noindent\textbf{Communication efficiency.} Fig. \ref{fig:subprotocol} illustrates the communication of each sub-protocol. 
Specifically, the left figure reports the communication of $\mathsf{secEval}$ under predicates with type $=,<,\ll$ and the number of candidate vertices $\in[1000,10000]$. 
It is note that the communication of predicates with type $=$ and $<$ is identical since both of them require {\csa} to communicate \textit{one bit} (i.e., re-share $\llbracket x_{\mathtt{V}_{c}}\rrbracket$) for each candidate vertex $\mathtt{V}_{c}$, but the communication of predicates with type $\ll$ is $2\times$ that of predicates with type $=,<$ since predicates with type $\ll$ consist of two predicates with type $<$.
The middle figure reports the communication of $\mathsf{secFetch}$ under two cases and the number of candidate vertices $\in[1000,10000]$. 
It is noted that the communication of case I is kept invariable, irrespective of the number of candidate vertices. 
That is because case I mainly requires local computation, and {\csa} only need to re-share two secret-shared vectors to achieve them in RSS (recall Algorithm \ref{alg:reveal}).
The right figure shows the communication of $\mathsf{secAccess}$ under different number of neighboring vertices.

\vspace{-10pt}

\subsection{Evaluation on Query Latency}

We now report the query latency, namely, given a query token, how long it takes {\csa} to obliviously execute subgraph matching on the encrypted attributed graph and output encrypted matching results. 
In particular, we first report the overall latency for different queries, and then report the breakdown of the  overall query latency.

\begin{table}[!t]
\centering
\small
\caption{Query Latency (s) under Different Values of $k$ ($k$-automorphism) and $|q|$ (Number of Target Vertices)} 		 	
\label{tab:latency}
\begin{tabular}{c|ccc|ccc|ccc}
	\hline
	&  \multicolumn{3}{c|}{Equality (=)} &  \multicolumn{3}{c|}{Range ($<$)}&  \multicolumn{3}{c}{Range ($\ll$)} \\\hline
	$k\lceil |q|$&$2$&$4$&$8$&$2$&$4$&$8$&$2$&$4$&$8$\\\hline
	2&0.7&2.1&2.5&1&3.4&4.0&1.5&4.1&4.5\\
	4&1.2&3.0&3.2&1.4&4.5&5.3&1.8&4.4&4.9\\
	6&1.5&3.3&3.9&2.1&7.4&7.6&3.5&6.1&7.5\\		\hline
\end{tabular}
\vspace{-10pt}
\end{table}

\begin{figure}[t!]
\centering
\begin{minipage}[t]{0.32\linewidth}
	\centering{\includegraphics[width=\linewidth]{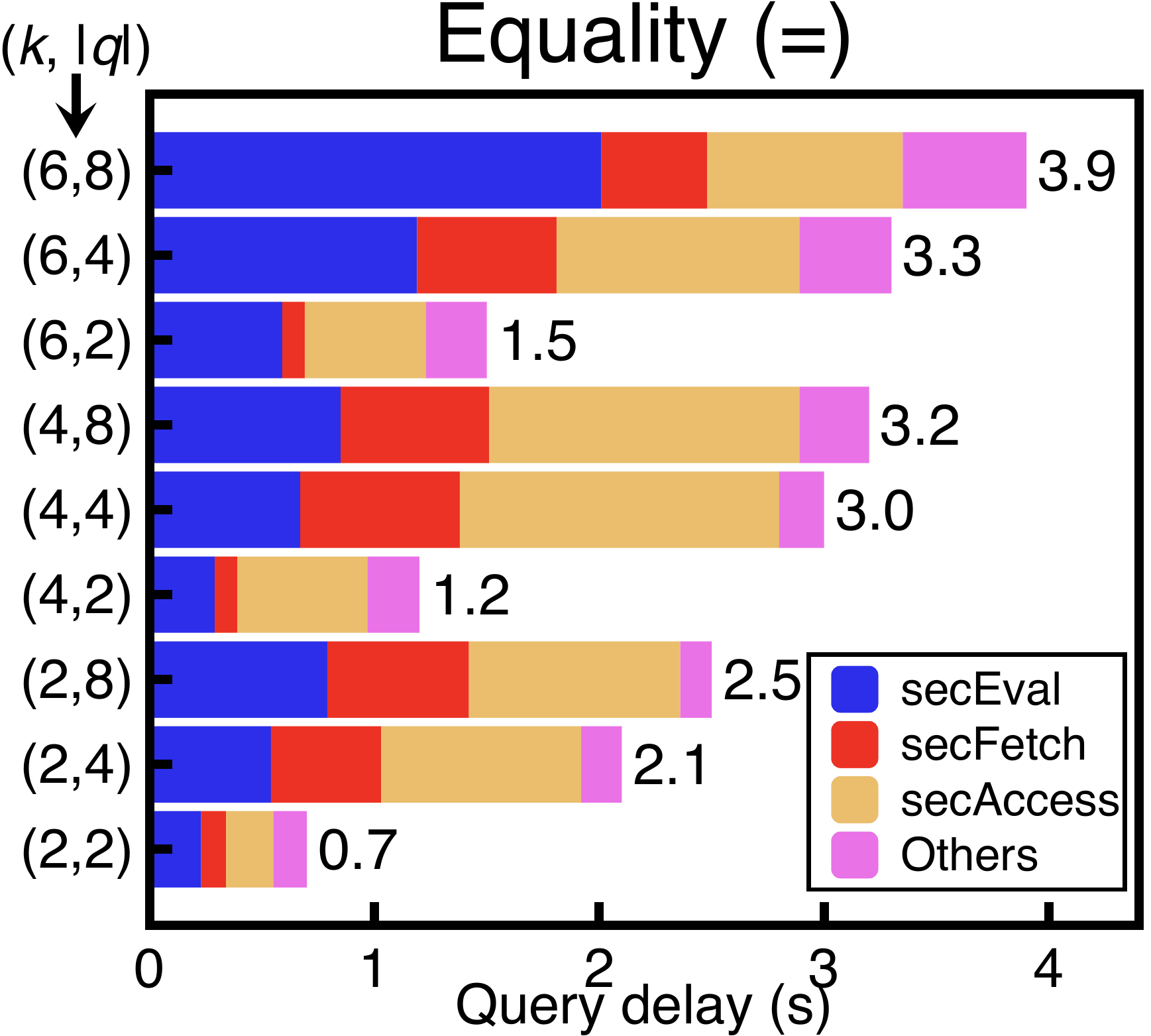}}
\end{minipage}
\begin{minipage}[t]{0.32\linewidth}
	\centering{\includegraphics[width=\linewidth]{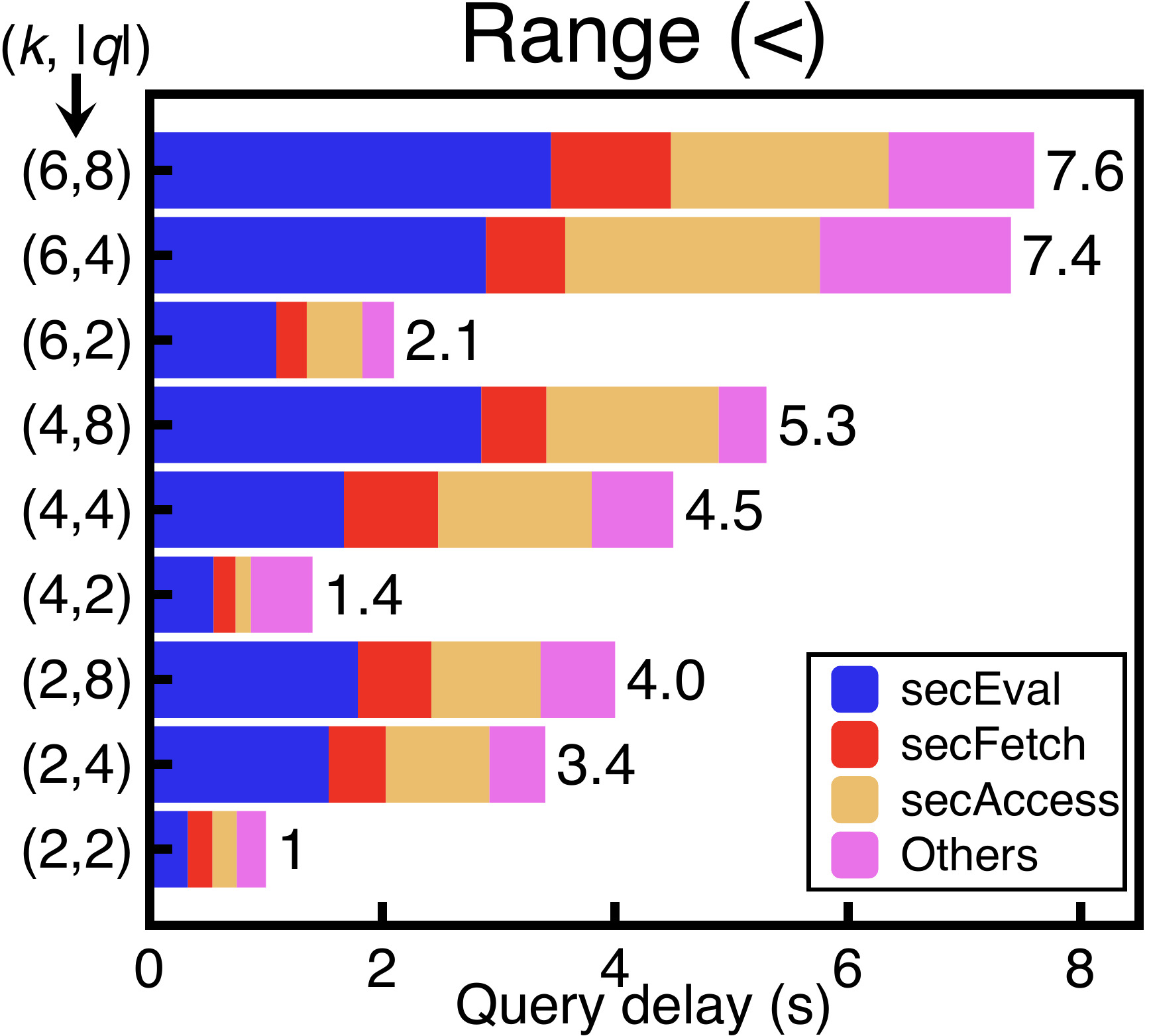}}
\end{minipage}
\begin{minipage}[t]{0.32\linewidth}
	\centering{\includegraphics[width=\linewidth]{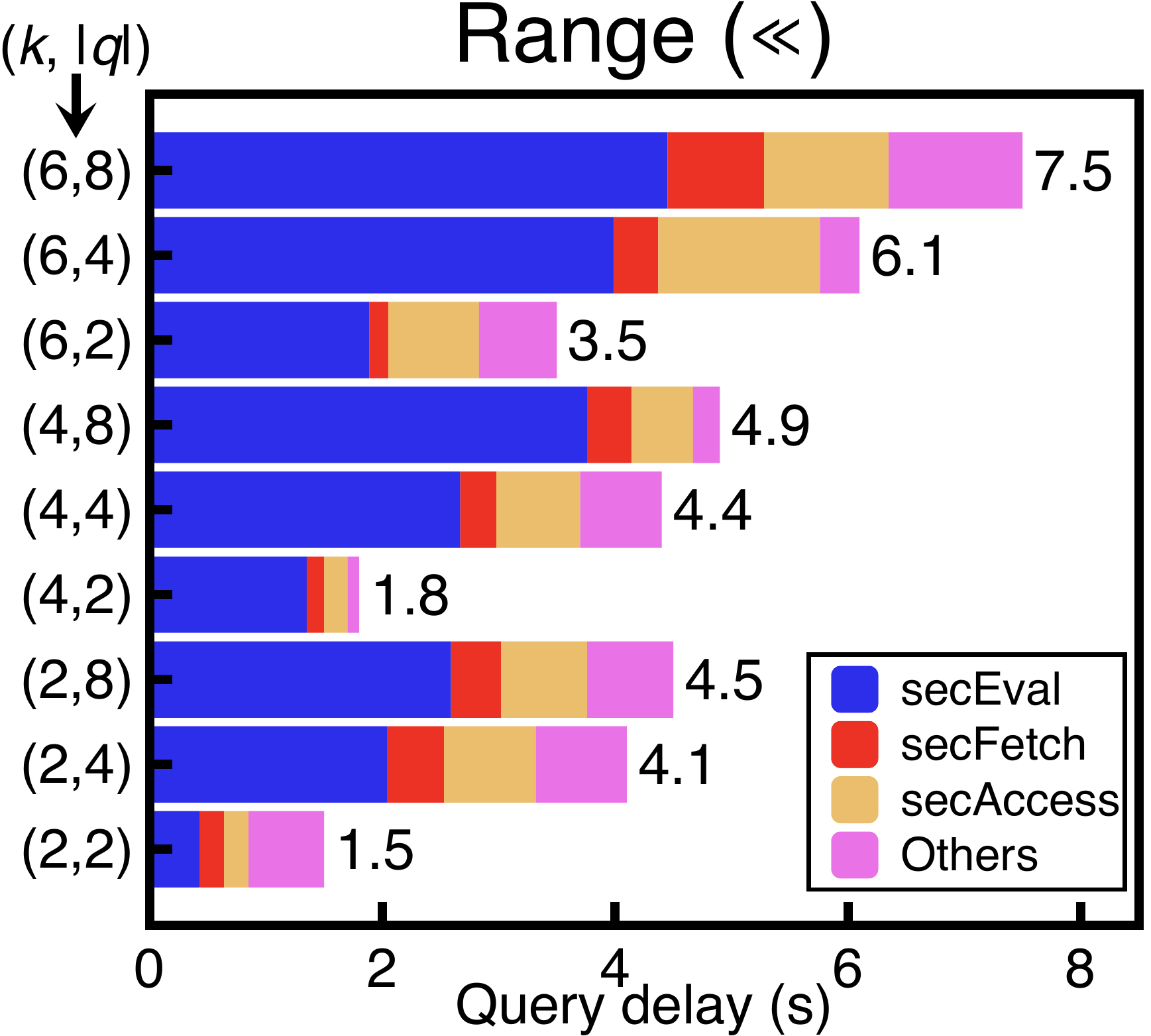}}
\end{minipage}
\caption{Breakdown of query latency (s) under different values of $k$ ($k$-automorphism) and $|q|$ (number of target vertices).}
\label{fig:lantency_breakdown}
\vspace{-10pt}
\end{figure}

\begin{table*}[!t]
\centering
\small
\caption{Bandwidth (MB) under Different $k$ ($k$-automorphism) and  $|q|$ (Number of Target Vertices)} \label{tab:comm}
\begin{tabular}{c|ccc|ccc|ccc}
	\hline
	&  \multicolumn{3}{c|}{Equality (=)} &  \multicolumn{3}{c|}{Range ($<$)}&  \multicolumn{3}{c}{Range ($\ll$)} \\\hline
	$k$&$|q|=2$&$|q|=4$&$|q|=8$&$|q|=2$&$|q|=4$&$|q|=8$&$|q|=2$&$|q|=4$&$|q|=8$\\\cline{2-10}
	2&64.74&87.15&103.65&72.4&94.15&112.9&66.5&77&90.6\\
	4&78.15&99.25&120.45&98.15&119.25&140.45&84.15&101.25&120.45\\
	6&91.24&117.34&137.53&111.2&137.3&147.5&97.2&120.3&139.5\\
	\hline
\end{tabular}
\vspace{-10pt}
\end{table*}

\noindent\textbf{Overall query latency.} For simplicity, we conduct experiment using 2-hop subgraph queries, with varying predicate types (i.e., $=,<$ and $\ll$), the number of target vertices (i.e., $|q|\in\{2,4,8\}$) and $k$-automorphism (i.e., $k\in\{2,4,6\}$), and summarize the experiment results in Table \ref{tab:latency}. 
It can be observed that the query latency is not linear in $|q|$. 
The reason is that different target vertices in the same hop can be evaluated in parallel, which enables us to allocate the secure predicate evaluation of each target vertex on independent GPU cores, and they work simultaneously.
To better understand the query latency of {\main}, we next report the breakdown of the overall query latency.

\noindent\textbf{Breakdown of query latency.} Fig. \ref{fig:lantency_breakdown} shows the breakdown of the overall query latency (as given in Table \ref{tab:latency}). 
It is observed that the majority of latency is due to $\mathsf{secEval}$ and $\mathsf{secAccess}$. 
The time cost of $\mathsf{secEval}$ is mainly due to local computation, because {\csa} must obliviously and locally evaluate the encrypted predicates on the encrypted attribute of each candidate vertex.
However, at the end of each evaluation, {\csa} only need to communicate one bit with each other.
In contrast, the time cost of $\mathsf{secAccess}$ is dominated by communication latency, because it requires {\csa} to obliviously shuffle the posting list of each matched vertex. 

% \vspace{-150pt}

%Therefore, for larger numbers of records, the computation and the bandwidth increase, but the number of round trips does not, and so the ratio of compute time to network time increases.
%It can be observed  that when the structure of the query is simple (e.g., $N=2$ and predicate $=$), the majority of latency is due to $\mathsf{secRank}$, but when the structure is complex (e.g., $N=8$ and predicate $\ll$), the majority of latency is due to $\mathsf{secEval}$ and $\mathsf{secAccess}$. 
%
%This is because the complex query structure leads to more predicate evaluation time but fewer matching results, and thus faster ranking process. 
%
%In addition, it can be observed that under different $k$, the same query has the same ranking time. 
%
%This is because {\main} produces the accurate subgraph matching and the dummy edges will not change the matching results.%, and thus the same query will output the same matching results, leading to the same ranking time.

\subsection{Evaluation on Server-Side Communication}
% We now report the server-side communication. 
%
% In particular, we first report the overall communication for different queries, and then report the breakdown of the overall communication.

\noindent\textbf{Overall communication.} We evaluate the same queries as that in the above experiments and provide the results in Table \ref{tab:comm}. 
It can be observed that similar to the above experiments, the communication is not linear in $|q|$. 
To better understand the communication of {\main}, we next report the breakdown of the overall communication.

\begin{figure}[!t]
\centering
\begin{minipage}[t]{0.32\linewidth}
	\centering{\includegraphics[width=\linewidth]{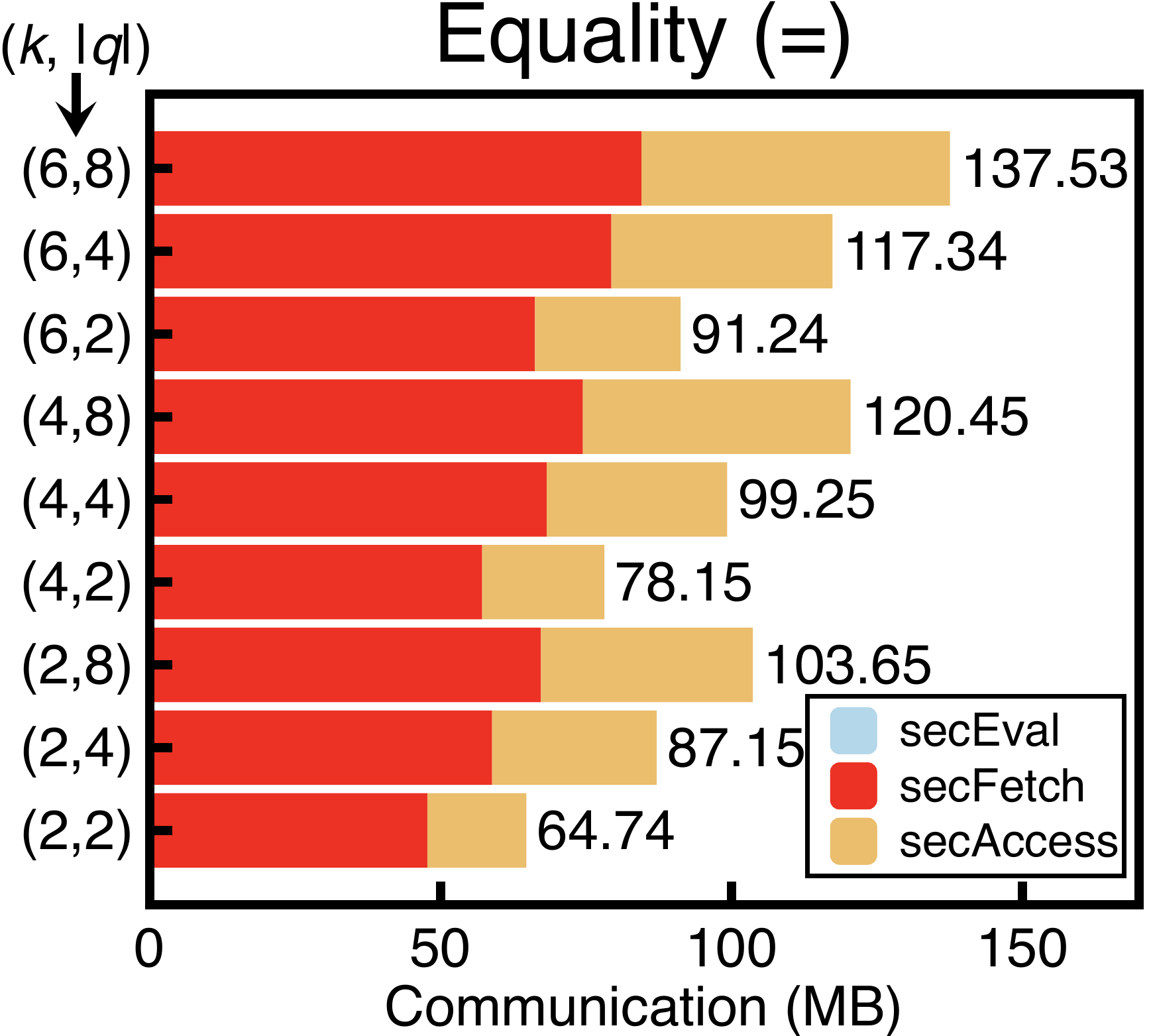}}
\end{minipage}
\begin{minipage}[t]{0.32\linewidth}
	\centering{\includegraphics[width=\linewidth]{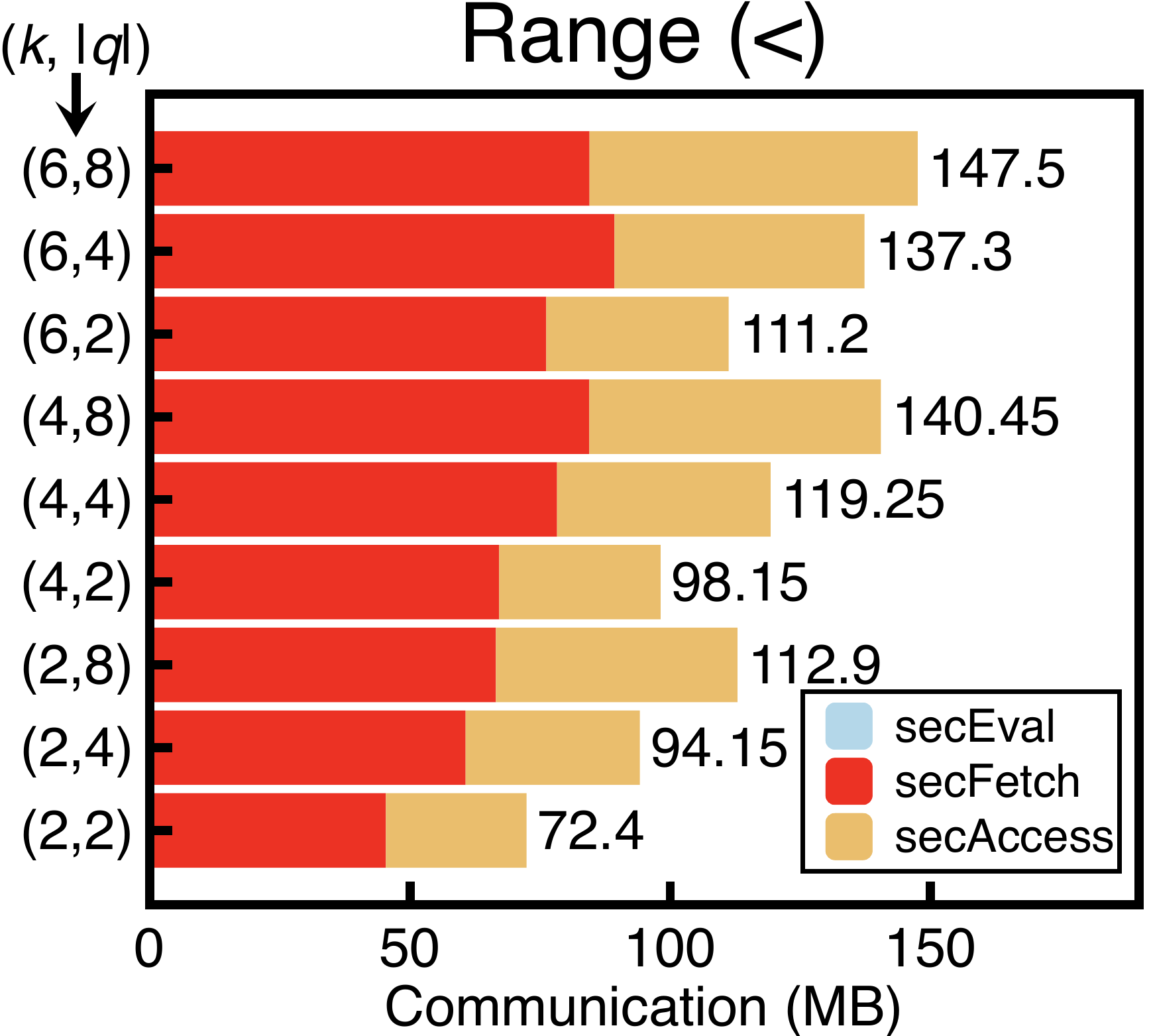}}
\end{minipage}	
\begin{minipage}[t]{0.32\linewidth}
	\centering{\includegraphics[width=\linewidth]{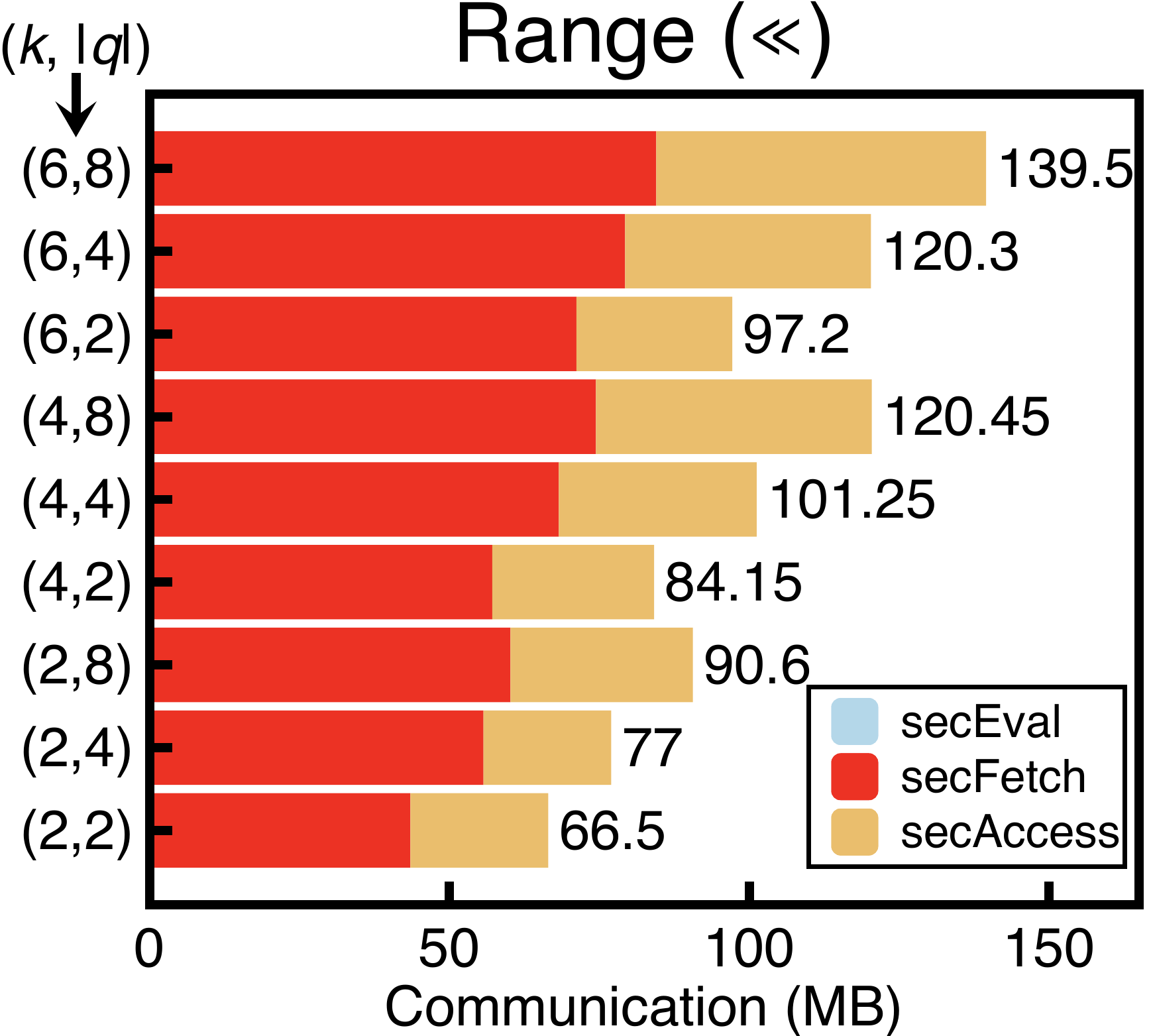}}
\end{minipage}
\caption{Breakdown of communication under different values of $k$ ($k$-automorphism) and $|q|$ (number of target vertices).}
\label{fig:comm_breakdown}
\vspace{-20pt}
\end{figure}

\noindent\textbf{Breakdown of communication.} Fig. \ref{fig:comm_breakdown} illustrates the breakdown of the overall communication. 
It is observed that the communication of $\mathsf{secEval}$ is inconspicuous in Fig. \ref{fig:comm_breakdown}. 
Because in $\mathsf{secEval}$, {\csa} only need to re-share \textit{one bit} (i.e., $\llbracket x_{\mathtt{V}_{c}} \rrbracket$) for each candidate vertex $\mathtt{V}_{c}$. 
The majority of communication is due to $\mathsf{secFetch}$ and $\mathsf{secAccess}$ since they require secure shuffle.
%
%The communication of $\mathsf{secRank}$ is mainly due to the secure evaluation of DCF. 

\section{Conclusion and Future Work}
\label{sec:conclusion}

We design, implement, and evaluate  {\main},  a new system enabling oblivious attributed subgraph matching services outsourced to the cloud, with stronger security and richer functionalities over prior art.
At the core of {\main} is a delicate synergy of attributed graph modelling and lightweight cryptographic techniques like FSS, RSS, and secret-shared shuffling. 
Extensive experiments over a real-world attributed graph dataset in the real cloud environment demonstrate that {\main} achieves practically affordable performance.

For future work, it would be interesting to explore how to extend our initial research effort to support oblivious attributed subgraph matching with malicious security.
%
%  provide secure attributed subgraph matching under malicious adversary model. 
% %
% Generally speaking, we can apply information-theoretic MACs to our {\main} such that {\cli} chooses a private value $\alpha$ and then receives the query results of the form $\{g_{m}, \sigma:=\alpha g_{m}\}$. 
% %
% Evaluating $\{\sigma=\alpha g_{m}\}$ allows {\cli} to check whether the cloud performed the computation honestly. 
%
Other directions for future work are to investigate the support for more complex scenarios such as dynamic graphs, as well as the possibility of leveraging the recent advances in trusted hardware for performance speedup.

\section*{Acknowledgments}

This work was supported in part by the Guangdong Basic and Applied Basic Research Foundation under Grant 2021A1515110027, in part by the Shenzhen Science and Technology Program under Grants RCBS20210609103056041 and JCYJ20210324132406016, in part by the National Natural Science Foundation of China under Grant 61732022, in part by the Guangdong Provincial Key Laboratory of Novel Security Intelligence Technologies under Grant 2022B1212010005, in part by the Research Grants Council of Hong Kong under Grants CityU 11217819, 11217620, RFS2122-1S04, N\_CityU139/21, C2004-21GF, R1012-21, and R6021-20F, and in part by the Shenzhen Municipality Science and Technology Innovation Commission under Grant SGDX20201103093004019.

\balance
\bibliographystyle{IEEEtran}
\bibliography{ref}

\end{document}